\documentclass[acmsmall,nonacm]{acmart}
\pdfoutput=1

\acmJournal{PACMPL}

\setcopyright{none}

\usepackage{amsmath}
\AtBeginDocument{\let\cite\citep}
\usepackage{mathtools}
\usepackage{stmaryrd}
\usepackage{array}
\usepackage{enumitem}
\usepackage{listings}
\usepackage{tabularx}
\usepackage{tikz}
\usetikzlibrary{arrows.meta, fit, backgrounds, calc, positioning, decorations.pathreplacing}
\usepackage{fontawesome5}
\microtypesetup{expansion=false}

\newcolumntype{L}[1]{>{\raggedright\arraybackslash}p{#1}}
\newcolumntype{Y}{>{\raggedright\arraybackslash}X}

\newcommand{\leanmonofont}{\ttfamily}

\definecolor{leanbg}{RGB}{248,248,250}
\definecolor{leanframe}{RGB}{220,223,230}
\definecolor{teal}{RGB}{0,128,128}
\definecolor{leankw}{RGB}{38,84,124}
\definecolor{leancomment}{RGB}{90,118,84}
\definecolor{leanstring}{RGB}{165,84,44}
\definecolor{leanident}{RGB}{56,56,56}
\definecolor{leanaccent}{RGB}{124,44,107}

\lstdefinelanguage{Lean4}{
  keywords={theorem, lemma, def, structure, inductive, class, instance, where,
    import, open, namespace, section, variable, abbrev, private, protected,
    scoped, attribute, local, omit, include, axiom, sorry, by, exact, simp,
    simpa, rw, rfl, cases, induction, apply, intro, ext, refine, constructor,
    have, let, show, calc, match, with, if, then, else, do, return, fun,
    noncomputable, termination_by, decreasing_by, Type, Prop, Sort},
  sensitive=true,
  morecomment=[l]{--},
  morecomment=[s]{/-}{-/},
  morestring=[b]",
  morestring=[b]',
}
\lstdefinestyle{leanbase}{
  language=Lean4,
  basicstyle=\leanmonofont\small\color{leanident},
  keywordstyle=\bfseries\color{leankw},
  commentstyle=\itshape\color{leancomment},
  stringstyle=\color{leanstring},
  rulecolor=\color{leanframe},
  backgroundcolor=\color{leanbg},
  numbers=left,
  numberstyle=\tiny\color{gray},
  numbersep=8pt,
  keepspaces=true,
  showstringspaces=false,
  breaklines=true,
  columns=flexible,
  escapeinside={(*@}{@*)},
  frame=single,
  framerule=0.4pt,
  framesep=6pt,
  xleftmargin=1.2em,
  xrightmargin=0.8em,
  aboveskip=0.8em,
  belowskip=0.8em,
  literate=
    {ℕ}{{$\mathbb{N}$}}1
    {ℂ}{{$\mathbb{C}$}}1
    {ℝ}{{$\mathbb{R}$}}1
    {ℤ}{{$\mathbb{Z}$}}1
    {ℚ}{{$\mathbb{Q}$}}1
    {��}{{\textbf{q}}}1
    {⊢}{{$\vdash$}}1
    {⊨}{{$\models$}}1
    {ₜ}{{$_{\mathrm{t}}$}}1
    {ₚ}{{$_{\mathrm{p}}$}}1
    {ᵢ}{{$_{i}$}}1
    {ₗ}{{$_{l}$}}1
    {ₘ}{{$_{m}$}}1
    {ₖ}{{$_{k}$}}1
    {ₙ}{{$_{n}$}}1
    {ᵥ}{{$_{v}$}}1
    {₀}{{$_{0}$}}1
    {₁}{{$_{1}$}}1
    {₂}{{$_{2}$}}1
    {₃}{{$_{3}$}}1
    {⁺}{{$^{+}$}}1
    {⁻}{{$^{-}$}}1
    {ⁿ}{{$^{n}$}}1
    {ᶜ}{{$^{\mathsf{c}}$}}1
    {σ}{{$\sigma$}}1
    {ρ}{{$\rho$}}1
    {δ}{{$\delta$}}1
    {φ}{{$\varphi$}}1
    {ψ}{{$\psi$}}1
    {α}{{$\alpha$}}1
    {β}{{$\beta$}}1
    {γ}{{$\gamma$}}1
    {ε}{{$\varepsilon$}}1
    {κ}{{$\kappa$}}1
    {ι}{{$\iota$}}1
    {λ}{{$\lambda$}}1
    {μ}{{$\mu$}}1
    {π}{{$\pi$}}1
    {Φ}{{$\Phi$}}1
    {Σ}{{$\Sigma$}}1
    {ℓ}{{$\ell$}}1
    {��}{{$\mathcal{N}$}}1
    {⟦}{{$\llbracket$}}1
    {⟧}{{$\rrbracket$}}1
    {⟪}{{$\langle\!\langle$}}1
    {⟫}{{$\rangle\!\rangle$}}1
    {⟨}{{$\langle$}}1
    {⟩}{{$\rangle$}}1
    {∣}{{$|$}}1
    {‖}{{$\|$}}1
    {∥}{{$\|$}}1
    {→}{{$\to$}}1
    {↔}{{$\leftrightarrow$}}1
    {↦}{{$\mapsto$}}1
    {⇒}{{$\Rightarrow$}}1
    {∀}{{$\forall$}}1
    {∃}{{$\exists$}}1
    {∞}{{$\infty$}}1
    {≠}{{$\neq$}}1
    {≤}{{$\leq$}}1
    {≥}{{$\geq$}}1
    {≃}{{$\simeq$}}1
    {∘}{{$\circ$}}1
    {∧}{{$\wedge$}}1
    {∨}{{$\vee$}}1
    {¬}{{$\neg$}}1
    {∅}{{$\emptyset$}}1
    {∈}{{$\in$}}1
    {∉}{{$\notin$}}1
    {⊆}{{$\subseteq$}}1
    {∑}{{$\sum$}}1
    {∏}{{$\prod$}}1
    {√}{{$\surd$}}1
    {×}{{$\times$}}1
    {⊗}{{$\otimes$}}1
    {⨂}{{$\bigotimes$}}1
    {⊤}{{$\top$}}1
    {⊥}{{$\bot$}}1
    {†}{{$\dagger$}}1
    {★}{{$\star$}}1
    {⋆}{{$\star$}}1
    {▸}{{$\blacktriangleright$}}1
    {•}{{$\bullet$}}1
    {·}{{$\cdot$}}1
    {…}{{$\dots$}}1
    {—}{{--}}1
    {q̄}{{$\bar{q}$}}1
    {̄}{{\textasciimacron}}1
}
\lstdefinestyle{leanexcerpt}{
  style=leanbase,
  basicstyle=\leanmonofont\footnotesize\color{leanident},
}
\lstdefinestyle{leanappendix}{
  style=leanbase,
  basicstyle=\leanmonofont\scriptsize\color{leanident},
}
\newcommand{\Tr}{\mathrm{Tr}}
\newcommand{\qhl}{QHL$^{+}$}
\newcommand{\qhlpar}{QHL$^{+}_{\mathit{par}}$}
\newcommand{\qhltot}{QHL$^{+}_{\mathit{tot}}$}
\newcommand{\ket}[1]{\lvert #1\rangle}
\newcommand{\bra}[1]{\langle #1\rvert}
\newcommand{\ketbra}[1]{\lvert #1\rangle\!\langle #1\rvert}
\newcommand{\inv}{\mathcal{J}}  %
\newcommand{\invL}{\mathcal{L}} %
\newcommand{\QFT}{\mathit{QFT}}
\newcommand{\IQFT}{\mathit{IQFT}}
\newcommand{\QPE}{\mathit{QPE}}
\newcommand{\PH}{\mathit{PH}}
\newcommand{\SWAP}{\mathit{SWAP}}
\newcommand{\CNOT}{\mathit{CNOT}}
\newcommand{\CZ}{\mathit{CZ}}
\newcommand{\cR}{\mathit{cR}}
\newcommand{\kw}[1]{\mathbf{#1}}
\newcommand{\cU}{\mathit{cU}}

\newcommand{\ctrl}[1]{\mathit{c}\text{-}#1}  %
\newcommand{\RUS}{\mathit{RUS}}
\newcommand{\LCU}{\mathit{LCU}}
\newcommand{\QSP}{\mathit{QSP}}
\newcommand{\Trotter}{\mathit{Trotter}}

\newcommand{\Ph}{\mathit{Ph}}
\newcommand{\phasegate}{\mathit{phase}}

\usepackage{xltabular}
\newlength{\pooutrule}
\newlength{\pooutpad}
\newlength{\pooutinner}%
\newcommand{\pooutgaprow}{%
  \multicolumn{3}{|@{}c@{}|}{\rule{0pt}{\pooutpad}}\\}
\newcommand{\poouttitle}{}%
\newcommand{\poouttitlerow}{%
  \multicolumn{3}{|@{}c@{}|}{%
    \rule[-5pt]{0pt}{17pt}\normalsize\bfseries\poouttitle}\\
  \noalign{\hrule height 0.4pt}}
\newenvironment{proofoutline}[2]{%
  \par\phantomsection\label{#1}%
  \gdef\poouttitle{#2}%
  \setlength{\LTpre}{\bigskipamount}%
  \setlength{\LTpost}{\bigskipamount}%
}{\ignorespacesafterend}

\newlength{\qindentlen}
\newcommand{\qindent}[1]{%
  \hspace*{\dimexpr #1\qindentlen\relax}%
}
\newcommand{\qlinenum}[1]{%
  \leavevmode\hspace*{\dimexpr-\pooutpad+0.5em\relax}{\footnotesize$\mathbf{#1}$}%
}
\newcommand{\qline}[4]{%
  \qlinenum{#1} &
  \qindent{#2}$#3$ &
  #4 \\[0.8pt]
}
\newcommand{\qassert}[3]{%
  {} &
  \qindent{#1}$\{\,#2\,\}$ &
  #3 \\[0.8pt]
}

\newenvironment{outlinerows}[1][\footnotesize]{%
  #1%
  \setlength{\arrayrulewidth}{\pooutrule}%
  \setlength{\pooutinner}{\dimexpr\linewidth-2\pooutpad-2\pooutrule\relax}%
  \xltabular{\linewidth}{%
    |@{\hspace{\pooutpad}}%
    >{\raggedright\arraybackslash}p{1.25em}%
    >{\raggedright\arraybackslash}X%
    >{\raggedright\arraybackslash}p{0.44\pooutinner}%
    @{\hspace{\pooutpad}}|}%
  \hline
  \poouttitlerow
  \pooutgaprow
  \endfirsthead
  \hline
  \multicolumn{3}{|@{\hspace{\pooutpad}}l@{\hspace{\pooutpad}}|}{%
    \rmfamily\footnotesize\itshape(continued)}\\
  \pooutgaprow
  \endhead
  \pooutgaprow
  \hline
  \endfoot
  \pooutgaprow
  \hline
  \endlastfoot
}{\endxltabular}

\title{QEDIA: Quantum Program Verification in Lean}

\author{Yuanjie (Collin) Ren}
\authornote{These authors contributed equally to this work.}
\affiliation{%
  \institution{Massachusetts Institute of Technology}
  \city{Cambridge}
  \state{MA}
  \country{USA}}
\affiliation{%
  \institution{Atomic Inference Inc}
  \city{Cambridge}
  \state{MA}
  \country{USA}}
\email{yuanjie@mit.edu}

\author{Zhicheng Zhang}
\authornotemark[1]
\affiliation{%
  \department{School of Information and Communication Technology}
  \institution{Griffith University}
  \city{Brisbane}
  \country{Australia}}
\email{iszczhang@gmail.com}

\author{Junyi Liu}
\authornotemark[1]
\affiliation{%
  \department{Joint Center for Quantum Information and Computer Science}
  \institution{University of Maryland}
  \city{College Park}
  \state{MD}
  \country{USA}}

\author{Jinzheng Li}
\affiliation{%
  \department{Department of Physics}
  \institution{Northeastern University}
  \city{Boston}
  \state{MA}
  \country{USA}}

\author{Yidi Qi}
\affiliation{%
  \institution{Harvard University}
  \city{Cambridge}
  \state{MA}
  \country{USA}}

\author{Mingsheng Ying}
\affiliation{%
  \department{Centre for Quantum Software and Information}
  \institution{University of Technology Sydney}
  \city{Sydney}
  \country{Australia}}
\email{Mingsheng.Ying@uts.edu.au}

\begin{document}

\begin{abstract}

Practical quantum programs interleave quantum and classical computation: classical data shape the quantum operations that are performed, and measurement outcomes shape the classical computation that follows. Formal verification of such programs in a proof assistant requires mechanizing a program logic that handles this interplay consistently at every level, from the mathematics of quantum mechanics, through program syntax and semantics, to specifications and proof rules. Such mechanizations, along with the program proofs built on them, typically require substantial manual effort.

We present QEDIA, a framework for verifying quantum programs with classical variables in Lean~4. At its foundation is an open library that formalizes the quantum Hoare logic \qhl{} across all these levels and proves the logic sound for both partial and total correctness. We construct the library and the program proofs, under human supervision, using an LLM-based prover agent, whose general-purpose core we extend for quantum program verification. We illustrate QEDIA with verified programs including the quantum Fourier transform, phase estimation, repeat-until-success circuits, and Hamiltonian simulation.
\end{abstract}
\maketitle
\tableofcontents
\section{Introduction}\label{sec:intro}

Classical software engineering relies on testing and debugging: running code,
inspecting intermediate states, and fixing bugs. For quantum programs, this
practice is constrained by features of quantum mechanics, notably: measurement
generally changes the state of a system, and the no-cloning
theorem~\cite{WoottersZurek1982,Dieks1982} forbids copying an unknown
state for offline analysis. Methods for assessing quantum operations, such
as quantum process
tomography~\cite{ChuangNielsen1997,MohseniRezakhaniLidar2008} and
randomized benchmarking~\cite{Knill2008RB}, are either expensive or report
only average behavior. Formal verification, which
proves by mathematical reasoning that a program satisfies its
specification, does not require running the program. It is widely used
for safety- and security-critical classical software and provides a
principled route to correctness guarantees for quantum programs.

The same physics that constrains testing and debugging also reshapes the
semantics a verifier must formalize. The programs considered here combine
classical and quantum variables and operations. The classical--quantum interplay in such programs runs in both directions:
classical variables determine the control flow, the gates applied, the qubits they act on, and their
parameters; conversely, a measurement writes its outcome into the classical
state, where it can influence the rest of the computation.  Reasoning about such programs is the
subject of a substantial body of work on quantum program
verification~\cite{Ying2012QHL,Liu2019QHLProver,Zhou2023,FengYing2021,Hietala2021SQIR,Chareton2021Qbricks,Li2024Qafny}.

Hoare logic is a cornerstone of classical program verification. Quantum Hoare logic~\cite{Ying2012QHL} was developed for purely quantum
programs and later formalized in Isabelle/HOL~\cite{Liu2019QHLProver} and
Coq~\cite{Zhou2023}; a subsequent extension~\cite{FengYing2021} added classical variables.
\qhl{}~\cite{Ying2025} extends this line of work with the features needed in
practice: quantum arrays, parameterized gates, richer predicate syntax, and
additional proof rules for reasoning about hybrid classical--quantum control.
Specifications take the form of quantum Hoare triples
$\{\varphi,A\}\,P\,\{\psi,B\}$, where the precondition $(\varphi,A)$ and the
postcondition $(\psi,B)$ each pair a classical formula with a quantum
predicate that may depend on classical variables.

To use \qhl{} for machine-checked program verification, we formalize it in
the Lean~4 proof assistant~\cite{deMoura2021Lean4}. The formalization spans the mathematics of
quantum mechanics, program syntax and semantics, and the proof rules, which
we prove sound. Mechanized verification of quantum programs, from the
underlying libraries to the program proofs, typically requires substantial
manual effort, as reflected in mature frameworks such as
CoqQ~\cite{Zhou2023} and Qbricks~\cite{Chareton2021Qbricks} and in the
certified end-to-end verification of Shor's algorithm~\cite{Peng2023Shor}.
Mechanizing \qhl{} is more demanding still, since the classical--quantum
interplay must be handled consistently at each level of the formalization, as the
examples in Section~\ref{sec:motivating} illustrate.

This proof-writing effort motivates our use of theorem-proving agents
that combine large language models (LLMs) with feedback from proof
assistants~\cite{baba2026proveragentagentbasedframework,varambally2026hilbertrecursivelybuildingformal,pmlr-v306-liu26gp,breen2026axproverdeepreasoningagentic,RenLiQi2026,li2026merleanproverrecursiveloopingharness}.
The general approach is to generate candidate formal proofs and revise
them using the proof assistant's diagnostics. In our work, LLM agents working under human supervision formalize the
quantum program logic, build its supporting library, and construct the example proofs, all checked by Lean's kernel. To our
knowledge, this is the first mechanized quantum program logic built by LLM
agents.

\subsection*{This Paper and Its Contributions}

This paper presents QEDIA, a framework for verifying quantum programs
with classical variables in Lean~4, based on \qhl{} and built on Mathlib~\cite{Mathlib2020}.\footnote{
\url{https://github.com/qedia/quantum-formal-verification}.}
Its contributions are:

\begin{itemize}[leftmargin=1.4em]
\item \textbf{A formalized quantum Hoare logic with proved soundness.}
  QEDIA formalizes \qhl{} for quantum while-programs with classical
  variables, quantum arrays, and parameterized gates. Its proof system is
  proved sound in Lean for both partial and total correctness.
  The formalization handles the classical--quantum interplay consistently at
  every level, from the mathematics of quantum mechanics to the proof system
  (Section~\ref{sec:faithful}).
\item \textbf{An agent-based proof-writing framework.}
  A general-purpose core agent coordinates LLM agents to construct and
  review formal proofs, combining Lean checks with review of the intended
  mathematics. Its quantum extension supports applying \qhl{} rules and
  adds quantum-specific checks (Section~\ref{sec:qpv}).
  We expect this workflow to be useful for future AI-assisted formal
  verification of quantum programs.
\item \textbf{A series of verified program examples.}
  We illustrate QEDIA with eleven verified programs: preparation of the
  Bell state and a 1D cluster state; parallel Hadamard; the quantum
  Fourier transform and its inverse; quantum phase estimation; a
  repeat-until-success circuit; Hamiltonian simulation of the
  transverse-field Ising model by a Trotter product formula; and the
  Bravyi--Kitaev encoding, block encoding, and quantum signal processing
  for the Fermi--Hubbard model (Section~\ref{sec:suite} and
  Appendix~\ref{app:outlines}).
  To the best of our knowledge, some of these examples cannot be directly
  handled by existing quantum program verifiers, since they involve the
  classical--quantum interplay at both the program and specification
  levels.
\end{itemize}

The remainder of the paper is organized as follows.
Section~\ref{sec:motivating} uses three motivating examples to identify the
levels at which the classical--quantum interplay must be handled.
Section~\ref{sec:logic} provides background on \qhl{}: its programs,
assertions, and proof rules. Section~\ref{sec:faithful} presents the Lean
formalization of \qhl{} and shows how it handles this interplay at each
level. Section~\ref{sec:qpv} describes the prover agent that constructs the
library and the program proofs. Section~\ref{sec:suite} is devoted to the verified program
examples, with detailed derivations for the examples of
Section~\ref{sec:motivating}. Section~\ref{sec:related} discusses related
work, and Section~\ref{sec:conclusion} concludes.
Appendix~\ref{app} gives further details of the formalization, and
Appendix~\ref{app:outlines} presents the remaining verified examples.

\section{Motivating Examples}\label{sec:motivating}
We now make the classical--quantum interplay described in Section~\ref{sec:intro} concrete. Through the three programs in Figure~\ref{fig:programs}, we illustrate the capabilities that a verifier needs in order to handle this interplay at different levels.

\begin{figure}[t]
\centering\footnotesize
\setlength{\tabcolsep}{0pt}
\begin{tabular}{@{}p{0.23\linewidth}@{\hspace{0.6em}}p{0.375\linewidth}@{\hspace{0.6em}}p{0.355\linewidth}@{}}
$\begin{array}{@{}l@{}}
m:=1;\ k:=0;\\
\kw{while}\ (m=1\wedge k<N)\ \kw{do}\\
\quad a:=\ket0;\\
\quad W_{\RUS}[a,d];\\
\quad m:=M[a];\\
\quad k:=k+1\\
\kw{od}.
\end{array}$
&
$\begin{array}{@{}l@{}}
i:=0;\\
\kw{while}\ i<n\ \kw{do}\\
\quad H[q[i]];\ j:=0;\\
\quad\kw{while}\ j<n-1-i\ \kw{do}\\
\quad\ \ \ \cR(2\pi/2^{j+2})[q[i],q[i+1+j]];\\
\quad\ \ \ j:=j+1\\
\quad\kw{od};\\
\quad i:=i+1\\
\kw{od};\\
i:=0;\\
\kw{while}\ i<\lfloor n/2\rfloor\ \kw{do}\\
\quad\SWAP[q[i],q[n-1-i]];\ i:=i+1\\
\kw{od}.
\end{array}$
&
$\begin{array}{@{}l@{}}
i:=0;\\
\kw{while}\ i<N\ \kw{do}\\
\quad j:=0;\\
\quad\kw{while}\ j<n\ \kw{do}\\
\quad\ \ \ R_x(2a[j]\delta)[q[j]];\ j:=j+1\\
\quad\kw{od};\\
\quad j:=0;\\
\quad\kw{while}\ j<n-1\ \kw{do}\\
\quad\ \ \ R_{zz}(2b[j]\delta)[q[j],q[j+1]];\\
\quad\ \ \ j:=j+1\\
\quad\kw{od};\\
\quad i:=i+1\\
\kw{od}.
\end{array}$
\\[5pt]
\centering\arraybackslash (a) $\RUS$ &
\centering\arraybackslash (b) $\QFT_n$ &
\centering\arraybackslash (c) $\Trotter_n$
\end{tabular}
\caption{The three motivating programs. Classical variables include loop
counters, measurement outcomes, and inputs; $q$ is a quantum array, and a quantum
operation acts on the qubits named in brackets.}
\label{fig:programs}
\end{figure}

\begin{example}[Repeat-Until-Success]\label{ssec:rus}
Our first example is a repeat-until-success (RUS) circuit. Such circuits use mid-circuit
measurements and classical feedback, in contrast to ordinary unitary synthesis, to reduce the number
of expensive $T$ gates in implementing certain non-Clifford unitary gates~\cite{PaetznickSvore2014,BocharovRoettelerSvore2015}.

The program in Figure~\ref{fig:programs}(a) presents an RUS implementation of the single-qubit unitary gate $V_{\RUS}=\frac{I+i\sqrt2 X}{\sqrt3}$, using a two-qubit gate $W_{\RUS}$ that is easier to construct from a few $T$ gates~\cite[Fig.~8]{PaetznickSvore2014}.
Here, $d$ is a data qubit, $a$ is an ancilla qubit, and $N$ is a classical input bounding the number of attempts. The variables $N,k$ are natural numbers, and the measurement outcome $m$ has type $\{0,1\}$.
The main component of the program is a $\mathbf{while}$ loop:
in each iteration we reset the ancilla qubit $a$ to $\ket{0}$,
apply gate $W_{\RUS}$, measure $a$ by the computational basis measurement $\{\ketbra 0,\ketbra 1\}$, and assign the measurement outcome to the classical variable $m$.
Given the specific $W_{\RUS}$, each iteration applies $V_{\RUS}$ to data qubit $d$ with probability $3/4$ (indicated by $m=0$)
and leaves $d$ unchanged with probability $1/4$ (indicated by $m=1$).
The loop stops when $m=0$ or the number of attempts reaches $N$.

The correctness of the RUS program in Figure~\ref{fig:programs}(a) consists of two parts:
the program terminates after applying $V_{\RUS}$ to $d$ or exhausting the $N$ attempts, and it applies $V_{\RUS}$ with probability bounded below in terms of $N$.
Using the \qhl{} framework, these two correctness conditions can be conveniently formalized as the following two Hoare triples:
\begin{equation}\label{eq:rustriple}
  \models_{\mathit{tot}}
  \bigl\{N\ge1,\ I\otimes\ketbra{\psi_0}\bigr\}\
  \RUS\
  \bigl\{m=0\vee k\ge N,\ \inv(m)\bigr\}
\end{equation}
and
\begin{equation}\label{eq:russuccess}
  \models_{\mathit{tot}}
  \bigl\{N\ge1,\ (1-4^{-N})(I\otimes\ketbra{\psi_0})\bigr\}\
  \RUS\
  \bigl\{m=0,\ \ketbra w\bigr\}.
\end{equation}
Here, $\ket{\psi_0}$ is an arbitrary initial state of $d$, 
\begin{equation}
 \label{eq:rus-invariant}   
\begin{gathered}
 \inv(m)=
 \begin{cases}
  \ketbra w,&m=0,\\
  I\otimes\ketbra{\psi_0},&m=1,
 \end{cases}\qquad\textup{and}\qquad
 \ket w=\ket0\otimes V_{\RUS}\ket{\psi_0}.
\end{gathered}
\end{equation}
We omit a classical precondition or postcondition when it is $\mathit{true}$.

Intuitively, the first Hoare triple~\eqref{eq:rustriple} guarantees that, with $d$ initially in $\ket{\psi_0}$, the program will
terminate and it either applies $V_{\RUS}$ successfully or exhausts the
$N$ attempts and leaves $d$ in $\ket{\psi_0}$.
The second Hoare triple~\eqref{eq:russuccess} guarantees that the overall success probability in implementing the desired $V_{\RUS}$ is at least $1-4^{-N}$.

\end{example}

Through the above example, we can see the classical--quantum interplay at both the program and specification levels.
For the program, the classical variable $m$ controls the $\mathbf{while}$ loop, yet its value in each iteration is set by measuring $a$. For the specification, in Hoare triples~\eqref{eq:rustriple} and~\eqref{eq:russuccess}, the same classical variable $m$ selects the successful branch and parameterizes the quantum predicate $\inv(m)$.

\begin{example}[Quantum Fourier Transform]\label{ssec:qft}
Our second example is the quantum Fourier transform (QFT). The QFT is a key subroutine of Shor's factoring algorithm~\cite{Shor1997} and, through its inverse, the readout stage of quantum phase estimation~\cite{Kitaev1995,Cleve1998}.

The program in Figure~\ref{fig:programs}(b) presents a standard QFT implementation~\cite{Coppersmith1994} on a quantum array $q$ of $n$ qubits.
Here, $q[j]$ is the $j$-th qubit of array $q$, and we define the controlled phase gate $\cR(\gamma)\,\ket{a,b} = e^{\,i\gamma ab}\,\ket{a,b}$ for $a,b\in\{0,1\}$.

The main component of the program is a nested $\mathbf{while}$ loop:
for each $i$, we first apply the Hadamard gate $H$ to qubit $q[i]$. Then, for each $j$, we apply $\cR(2\pi/2^{j+2})$ to qubits $q[i],q[i+1+j]$, where both the qubit indices and the phase $2\pi/2^{j+2}$ are parameterized by classical variables $i$ and $j$.
After these produce the Fourier phases, a final loop of $\SWAP$ gates reverses the qubit order.

The correctness of the QFT program in Figure~\ref{fig:programs}(b) states that, for each fixed $n\ge1$, the program terminates and maps computational basis states to the corresponding Fourier basis states.
This can be formalized as the following Hoare triple:
\begin{equation}\label{eq:qfttriple}
  \models_{\mathit{tot}}
  \bigl\{\ketbra{x}\bigr\}\
  \QFT_n\
  \bigl\{\ketbra{\chi(n,x)}\bigr\}.
\end{equation}
Here, the classical input $x$ is an array of bits representing the integer $\bar x=\sum_{j=0}^{n-1}x[j]\,2^{n-1-j}$ in binary, $\ket{x}=\bigotimes_{i=0}^{n-1}\ket{x[i]}$, and
\[
  \ket{\chi(n,x)}
  =\frac{1}{\sqrt{2^{\,n}}}\sum_{y\in\{0,1\}^n}
    e^{\,2\pi i\,\bar{x}\,\bar{y}/2^n}\,\ket{y}
\]
is the Fourier basis state.

Intuitively, the Hoare triple~\eqref{eq:qfttriple} guarantees that, with $q$ initially in an arbitrary computational basis state $\ket{x}$, the program will terminate and leave $q$ in the corresponding Fourier basis state $\ket{\chi(n,x)}$.
We further verify the QFT on arbitrary inputs in Section~\ref{ssec:ex-qft}.
\end{example}

The above example also illustrates the classical--quantum interplay at the program and specification levels. For the program, the classical variables $i,j$ determine both the qubit indices and the phase of the controlled phase gate. For the specification, the quantum predicates $\ketbra x$ and $\ketbra {\chi(n,x)}$ both depend on the array $x$.

\begin{example}[The Trotter Product-Formula Circuit]\label{ssec:trotter}
Our third example is a Trotter product-formula circuit for Hamiltonian simulation. Such circuits approximate the quantum evolution generated by a sum of Hamiltonians~\cite{Trotter1959,Lloyd1996}. 

The program in Figure~\ref{fig:programs}(c) presents a first-order Trotter implementation of the evolution $U_\theta=e^{-i\theta H}$ governed by the transverse-field Ising Hamiltonian
\[
 H=H_X+H_{ZZ},\qquad
 H_X=\sum_{j=0}^{n-1}a[j]X_j,\qquad
 H_{ZZ}=\sum_{j=0}^{n-2}b[j]Z_jZ_{j+1}.
\]
Here, $n\ge1$ is fixed, while the coefficient arrays $a,b$, evolution time $\theta$, and number $N$ of Trotter steps are classical inputs. The operators $X_j,Z_j$ are the Pauli operators on qubit $q[j]$.
The main component of the program is a $\mathbf{while}$ loop:
each iteration applies the single-qubit gate $R_x(2a[j]\delta)$, with $R_x(\gamma)=e^{-i\gamma X/2}$ and $\delta=\theta/N$, to every qubit $q[j]$. It then applies the two-qubit gate $R_{zz}(2b[j]\delta)$, with $R_{zz}(\gamma)=e^{-i\gamma(Z\otimes Z)/2}$, to every $q[j],q[j+1]$.
The loop runs $N$ iterations.

The correctness of the Trotter program in Figure~\ref{fig:programs}(c) is an approximation guarantee: the program terminates, and its output approximates $U_\theta$ applied to the input up to a target error.
This correctness condition can be formalized as the following Hoare triple:
\begin{equation}\label{eq:trottertriple}
 \models_{\mathit{tot}}
 \left\{N\ge1\wedge\varepsilon\ge\frac{\theta^2}{2N}\|[H_X,H_{ZZ}]\|,\
             \frac{\ketbra\psi}{1+\varepsilon}\right\}\
 \Trotter_n\
 \left\{\frac{\ketbra{U_\theta\psi}+\varepsilon I}{1+\varepsilon}\right\}.
\end{equation}
Here, $\ket\psi$ is an arbitrary input state of $q$, the classical input $\varepsilon\ge0$ is the target error, $[H_X,H_{ZZ}]=H_XH_{ZZ}-H_{ZZ}H_X$, and $\|\cdot\|$ is the operator norm. The classical precondition requires $N\ge1$ and that $\varepsilon$ be at least the first-order Trotter error bound $\frac{\theta^2}{2N}\|[H_X,H_{ZZ}]\|$~\cite{Childs2021TrotterError}.

Intuitively, the Hoare triple~\eqref{eq:trottertriple} guarantees that, for every $\varepsilon$ satisfying the classical precondition, the program will terminate and produce an output state with fidelity at least $1-\varepsilon$ with the target state $U_\theta\ket\psi$ under the desired evolution $U_\theta$.
\end{example}

Through the above example, we can see that the classical--quantum interplay enables us to express a correctness condition involving approximation. Showing exact implementation of the product formula alone does not establish the approximation guarantee of the Trotter circuit. Instead, in the Hoare triple~\eqref{eq:trottertriple}, the classical variable $\varepsilon$ represents the target error, and the quantum predicates are parameterized by $\varepsilon$. 

To sum up, these motivating examples illustrate the intricacy in handling the classical--quantum interplay at the following two levels:

\begin{itemize}[leftmargin=1.4em,itemsep=2pt,topsep=3pt]
\item \textbf{(L1) Program syntax and semantics.} 
    Quantum programs contain both classical and quantum variables. Classical states can affect quantum states by the guards of control flow, indices for arrays, or parameters for gates and measurements.
    Quantum states can affect classical states by the outcomes of measurements.
    A program may also be parameterized, as $\QFT_n$ and $\Trotter_n$ are by $n$.
\item \textbf{(L2) Specification and proof system.} 
    Building on \textbf{(L1)}, the specification of program correctness contains both classical and quantum parts. Quantum predicates can depend on classical variables.
    Consequently, the proof system has to deal with the classical--quantum interplay as well.
\end{itemize}
Existing quantum verifiers such as \cite{Chareton2021Qbricks,Hietala2021SQIR,Li2024Qafny,Liu2019QHLProver,Paykin2017QWIRE,Unruh2019qrhl,wu2025localreasoningprobabilisticbehaviour,Zhou2023} do not directly support the classical--quantum interplay at both levels: they admit the classical and quantum components only in a restricted combination, and none of them directly accommodates all of the specifications \eqref{eq:rustriple}--\eqref{eq:trottertriple}. A detailed comparison is presented in Section~\ref{sec:related}.
\qhl{} combines classical and quantum components as follows: classical variables, expressions, control flow, and classical formulas in specifications keep the same roles as in classical programs, and the quantum components are built around them.
However, mechanizing such a framework as a Lean library requires more than encoding the proof rules.
We further need to handle the classical--quantum interplay at the following level:

\begin{itemize}[leftmargin=1.4em,itemsep=2pt,topsep=3pt]
    \item \textbf{(L0) Mathematics and quantum mechanics.}
    Formalizing the semantics of quantum programs requires supporting structures from mathematics and quantum mechanics, including Hilbert spaces, states, linear operators, tensor products, probabilities, variables, gates, and measurements. While some of these can build on the existing Mathlib in Lean, one also needs to account for new requirements imposed by the classical--quantum interplay at higher levels \textbf{(L1)} and \textbf{(L2)}. Conversely, the design choices at this level determine how conveniently the higher levels can be expressed.
\end{itemize}
A good mechanization must therefore handle the classical--quantum interplay at all three levels together. Section~\ref{sec:faithful} describes how QEDIA does so.

\section{Preliminaries: \texorpdfstring{\qhl{}}{QHL+}}\label{sec:logic}

In this section, we briefly introduce the background on \qhl{}~\cite{Ying2025}.
This includes quantum programs, their correctness specification, and its proof system.

\subsection{Classical and Quantum Variables}\label{ssec:objects}

\paragraph{Classical variables.}
A classical variable has a type, such as
integers or booleans. A \emph{classical state} $\sigma$ assigns to each classical variable a value of the corresponding type. Evaluating a classical expression $e$ gives its value
$\sigma(e)$. We use $\sigma\models\varphi$ to denote that the classical
formula $\varphi$ holds in $\sigma$. 

\paragraph{Simple quantum variables.}
A simple quantum variable has a Hilbert space as its \emph{quantum type}. The state space of a set $S$ of several quantum variables is the tensor product of their Hilbert
spaces, denoted by $ \mathcal H_S=\bigotimes_{q\in S}\mathcal H_q$.
The joint state space of all quantum variables is denoted by $\mathcal H_{\mathit{all}}$.
An operator on a subsystem acts on a larger system by tensoring it with
the identity on the remaining quantum variables.

\paragraph{Quantum arrays and subscripted quantum variables.}
An array quantum variable of type $J\to\mathcal H_q$ consists of quantum
variables indexed by a classical type $J$, each with state space
$\mathcal H_q$. Its joint state space is $\bigotimes_{j\in J}\mathcal H_q$.
If $e$ is a classical expression of type $J$, the \emph{subscripted quantum
variable} $q[e]$ denotes the quantum variable $q[\sigma(e)]$ in classical
state $\sigma$. We write $\sigma(q[e])=q[\sigma(e)]$. For a simple
quantum variable, $\sigma(q)=q$.

For example, the Trotter circuit (Example~\ref{ssec:trotter}) applies gates to $q[j]$ and
$q[j+1]$, where the indices are parameterized by the classical variable $j$. For
a sequence $\bar q=(q_1,\ldots,q_k)$ of quantum variables, an operation acting on it should require
\begin{equation}
    \label{eq:dist}
 \sigma\models\mathit{Dist}(\bar q),
 \quad\textup{defined by}\quad
 \sigma(q_i)\ne\sigma(q_j)\ \text{for every }i\ne j.
\end{equation}
For example, $\CNOT[q[i],q[j]]$ requires $\sigma(i)\ne\sigma(j)$.

\subsection{Syntax and Semantics of \texorpdfstring{\textsf{qWhile}$^+$}{qWhile+}}

\paragraph{Syntax.}
Figure~\ref{fig:prog-syntax} gives the syntax of \textsf{qWhile}$^+$.
Most components are self-explanatory. Here, we emphasize the interplay between classical and quantum parts.
As usual, the guards in $\mathbf{if}$ and $\mathbf{while}$ statements are classical. In $U(\bar t)[\bar q]$, the parameters $\overline{t}$ and indices for subscripted quantum variables in $\bar q$ are also classical.
In the other direction, a measurement $x:=M[\bar q]$ assigns its outcome value to a classical variable. That value can enter subsequent expressions: in the RUS program of Example~\ref{ssec:rus}, the outcome $m$ appears in the loop condition $m=1\wedge k<N$.
We write $\bar q:=\ket0$ for $q_1:=\ket0;\cdots;q_k:=\ket0$.

\begin{figure}[t]
\[
\begin{aligned}
 P ::= {}&\mathbf{skip}\mid x:=e\mid q:=\ket0
          \mid U(\bar t)[\bar q]\mid x:=M[\bar q]\\
 &\mid P_1;P_2
  \mid\mathbf{if}\ b\ \mathbf{then}\ P_1\ \mathbf{else}\ P_0
  \mid\mathbf{while}\ b\ \mathbf{do}\ P.
\end{aligned}
\]
\caption{Syntax of \textsf{qWhile}$^+$. Here, $x$ is a classical variable, $e$ and
$\bar t$ are classical expressions, $b$ is a classical boolean expression,
$q$ is a quantum variable, $\overline{q}$ is a sequence of quantum variables, $U(\overline{t})$ is a unitary gate parameterized by $\overline{t}$, and $M$ is a measurement.}
\label{fig:prog-syntax}
\end{figure}

\paragraph{Operational semantics.}
We define a \emph{classical--quantum state}, or \emph{cq-state}, to be a pair
$(\sigma,\rho)$, where $\sigma$ is a classical state and $\rho$ is a partial density operator, that is, $\rho\ge0$ and
$\Tr(\rho)\le1$. The quantity $\Tr(\rho)$ is the probability carried by
this cq-state. 
This treatment makes it convenient to describe measurements in the operational semantics.

A configuration is $(P,\sigma,\rho)$, where $P$ is a remaining program to be executed or the termination symbol $\downarrow$.
The operational semantics of programs in Figure~\ref{fig:operational-rules} defines rules for the transition $\longrightarrow$ between configurations.

\begin{figure}[t]
\small
\[
\begin{gathered}
 (\mathbf{skip},\sigma,\rho)\longrightarrow(\downarrow,\sigma,\rho),
 \qquad
 (x:=e,\sigma,\rho)\longrightarrow
       (\downarrow,\sigma[x:=\sigma(e)],\rho),\\[3pt]
 (q:=\ket0,\sigma,\rho)\longrightarrow
       (\downarrow,\sigma,\mathcal E_{\sigma(q)}(\rho)),\qquad
 \frac{\sigma\models\mathit{Dist}(\bar q)}
 {(U(\bar t)[\bar q],\sigma,\rho)\longrightarrow
       (\downarrow,\sigma,U_\sigma\rho U_\sigma^\dagger)},\\[4pt]
 \frac{\sigma\models\mathit{Dist}(\bar q)}
 {(x:=M[\bar q],\sigma,\rho)\longrightarrow
       (\downarrow,\sigma[x:=m],M_m^\sigma\rho(M_m^\sigma)^\dagger)}
 \quad\text{for every outcome }m,\\[4pt]
 \frac{(P_1,\sigma,\rho)\longrightarrow(P'_1,\sigma',\rho')}
      {(P_1;P_2,\sigma,\rho)\longrightarrow(P'_1;P_2,\sigma',\rho')},
 \qquad\downarrow;P_2=P_2,\\[4pt]
 \frac{\sigma\models b}
 {(\mathbf{if}\ b\ \mathbf{then}\ P_1\ \mathbf{else}\ P_0,\sigma,\rho)
       \longrightarrow(P_1,\sigma,\rho)},\qquad
 \frac{\sigma\models\neg b}
 {(\mathbf{if}\ b\ \mathbf{then}\ P_1\ \mathbf{else}\ P_0,\sigma,\rho)
       \longrightarrow(P_0,\sigma,\rho)},\\[4pt]
 \frac{\sigma\models b}
 {(\mathbf{while}\ b\ \mathbf{do}\ P,\sigma,\rho)
       \longrightarrow(P;\mathbf{while}\ b\ \mathbf{do}\ P,\sigma,\rho)},\qquad
 \frac{\sigma\models\neg b}
 {(\mathbf{while}\ b\ \mathbf{do}\ P,\sigma,\rho)
       \longrightarrow(\downarrow,\sigma,\rho)}.
\end{gathered}
\]
\caption{Operational transition rules of \textsf{qWhile}$^+$.
Here, the update $(\sigma[x:=v])(y)$ is defined as $v$ if $y=x$, and $\sigma(y)$ if $y\neq x$. The unitary operator $U_\sigma$ represents $U(\sigma(\bar t))$ acting on $\sigma(\bar q)$. Each measurement operator $M_m^\sigma$ corresponds to an outcome $m$ and also acts on $\sigma(\bar q)$. The initialization is defined as $\mathcal E_{\sigma(q)}(\rho)
 =\sum_j(\ket0\!\bra j)_{\sigma(q)}\,
       \rho\,(\ket j\!\bra0)_{\sigma(q)}$.
       The condition $\mathit{Dist}$ is as defined in \eqref{eq:dist}.
}

\label{fig:operational-rules}
\end{figure}

\paragraph{Denotational semantics and termination.}
\phantomsection\label{def:source-denotation}
For an input cq-state $(\sigma,\rho)$, the denotational semantics can be defined directly via the operational semantics:
\[
 \llbracket P\rrbracket(\sigma,\rho)
 =\left\{\left|(\sigma',\rho')\ \middle|\
 (P,\sigma,\rho)\longrightarrow^*(\downarrow,\sigma',\rho')
 \right|\right\},
\]
where $\longrightarrow^*$ is the reflexive transitive closure of $\longrightarrow$ and $\{\!|\cdot|\!\}$ denotes a multiset. 
Here, a multiset is used because there may be transitions from different sequences of measurement outcomes terminating in the same final cq-state.
The total probability of terminating outputs satisfies
\[
 \sum_{(\sigma',\rho')\in\llbracket P\rrbracket(\sigma,\rho)}\Tr(\rho')
 \le\Tr(\rho).
\]
The non-termination probability starting from the cq-state $(\sigma,\rho)$ is then
\[
 \mathrm{NT}(P)(\sigma,\rho)
 =\Tr(\rho)-
   \sum_{(\sigma',\rho')\in\llbracket P\rrbracket(\sigma,\rho)}\Tr(\rho').
\]

\subsection{Formal Quantum States and Assertions}

\paragraph{Formal quantum states.}
In Section~\ref{sec:motivating}, we have seen Hoare triples parameterized by classical variables. For example, $\ketbra {\chi(n,x)}$ depends on the array $x$ for fixed $n$ in Example~\ref{ssec:qft}.
To describe quantum states like $\ket{\chi(n,x)}$, we need the notion of formal (pure) quantum states. Figure~\ref{fig:state-syntax} gives the syntax and signature of a formal quantum state, where the signature records the quantum variables it involves and is recursively defined.
Figure~\ref{fig:state-interpretation} specifies how formal quantum states are interpreted in a classical state.

\begin{figure}[t]
\small
\raggedright
\[
 s ::= \ket e_q\mid s_1\otimes s_2
       \mid\alpha_1s_1+\alpha_2s_2\mid (U(\bar t)[\bar q])s.
\]
\[
\begin{aligned}
 \mathit{Sig}(\ket e_q)&=\{q\},\\
 \mathit{Sig}(s_1\otimes s_2)&=\mathit{Sig}(s_1)\cup\mathit{Sig}(s_2),\\
 \mathit{Sig}(\alpha_1s_1+\alpha_2s_2)&=\mathit{Sig}(s_1)
   &&\text{if }\mathit{Sig}(s_1)=\mathit{Sig}(s_2),\\
 \mathit{Sig}((U(\bar t)[\bar q])s)&=\mathit{Sig}(s)
   &&\text{if }\bar q\subseteq\mathit{Sig}(s).
\end{aligned}
\]
\caption{Syntax and signatures of formal quantum states.
Here, $q$ is a quantum variable of type $\mathcal H_T$, $e$ is a classical expression of type $T$, and
$\alpha_1,\alpha_2$ are complex-valued classical expressions.
The unitary gate $U(\overline{t})[\overline{q}]$ is as defined in Figure~\ref{fig:prog-syntax}.
}
\label{fig:state-syntax}
\end{figure}

\begin{figure}[t]
\small
\raggedright

\[
\begin{aligned}
 \sigma(\ket e_q)&=\ket{\sigma(e)}_{\sigma(q)},\\
 \sigma(s_1\otimes s_2)&=\sigma(s_1)\otimes\sigma(s_2)
  &&\text{if }\sigma(\mathit{Sig}(s_1))\cap\sigma(\mathit{Sig}(s_2))=\emptyset,\\
 \sigma(\alpha_1s_1+\alpha_2s_2)
  &=\sigma(\alpha_1)\sigma(s_1)+\sigma(\alpha_2)\sigma(s_2),\\
 \sigma\bigl((U(\bar t)[\bar q])s\bigr)
  &=U_{\sigma} (\sigma(s))
  &&\text{if }\sigma\models\mathit{Dist}(\bar q).
\end{aligned}
\]
\caption{Interpretation of formal quantum states in a classical state $\sigma$. The unitary operator $U_{\sigma}$ is as defined in Figure~\ref{fig:operational-rules}. The condition $\mathit{Dist}$ is as defined in \eqref{eq:dist}.}
\label{fig:state-interpretation}
\end{figure}

\paragraph{Quantum predicate formulas.}
A \emph{quantum predicate} is a Hermitian operator $E$ satisfying
$0\le E\le I$, where $\le$ is the L\"owner order.
On a partial density operator $\rho$, its expectation satisfies
$0\le\Tr(E\rho)\le\Tr(\rho)$.
Quantum predicate formulas have the syntax in
Figure~\ref{fig:predicate-syntax}. 
Here, $K(\bar t)[\bar q]$ is an atomic formula with classical parameters $\bar t$, interpreted in a classical state $\sigma$ as the quantum predicate $K(\sigma(\bar t))$ on $\sigma(\bar q)$.
We also use the atomic formula $[s]$, where $s$ is a formal quantum state. There are three connectives for quantum predicate formulas:
negation $\neg$, tensor product $\otimes$, and a family $F(\bar t)=\{F_j(\bar t):j\}$ of Kraus operators acting on $\bar q$.

\begin{figure}[t]
\small
\raggedright
\[
 A ::= K(\bar t)[\bar q]\mid [s]\mid\neg A\mid A_1\otimes A_2
       \mid F(\bar t)[\bar q](\{A_j\}).
\]
\[
\begin{aligned}
 \mathit{Sig}(K(\bar t)[\bar q])&=\{\bar q\},
 &\mathit{Sig}([s])&=\mathit{Sig}(s),\\
 \mathit{Sig}(\neg A)&=\mathit{Sig}(A),
 &\mathit{Sig}(A_1\otimes A_2)&=\mathit{Sig}(A_1)\cup\mathit{Sig}(A_2),\\
 \mathit{Sig}(F(\bar t)[\bar q](\{A_j\}))&=X, &\textup{where }&\bar q\subseteq \mathit{Sig}(A_j)=X\textup{ for each $j$}.
\end{aligned}
\]
\caption{Syntax and signatures of quantum predicate formulas. Here, $K$ denotes a family of quantum predicates, and $F$ denotes a finite family of operators. Both are parameterized by $\bar t$ and act on $\bar q$. }
\label{fig:predicate-syntax}
\end{figure}

Figure~\ref{fig:predicate-interpretation} specifies how quantum predicate formulas are interpreted in a classical state.
We write $\sigma\models\mathrm{WD}(A)$ if $\sigma(A)$ is well-defined.
Most of the interpretations are self-explanatory.
The connective $F(\bar t)$ can express the following transformations on quantum predicates:
(1) reset, by taking $F_B=\{\ket{j}\!\bra{0}:j\}$; (2) unitary, by taking $F_U=\{U(\bar v)^\dagger\}$; (3) measurement, by taking $F_M=\{M_j^\dagger:j\}$ for a measurement $M=\{M_j:j\}$; and (4) masking, by taking $F_I(v)=\{I\}$ if $v$ is true and $\{0\}$ otherwise.
We write $b\cdot A$ for $F_I(b)[\bar q](A)$ with $\bar q= \emptyset$. 

\begin{figure}[t]
\small
\raggedright
\[
\begin{aligned}
 \sigma(K(\bar t)[\bar q])&=K(\sigma(\bar t))_{\sigma(\bar q)}
   &&\text{if }\sigma\models\mathit{Dist}(\bar q),\\
 \sigma([s])&=\ketbra{\sigma(s)}
   &&\text{if }\sigma(s)\text{ exists and }\|\sigma(s)\|=1,\\
 \sigma(\neg A)&=I-\sigma(A),\\
 \sigma(A_1\otimes A_2)&=\sigma(A_1)\otimes\sigma(A_2)
   &&\text{if }\sigma(\mathit{Sig}(A_1))\cap\sigma(\mathit{Sig}(A_2))=\emptyset,\\
 \sigma\bigl(F(\bar t)[\bar q](\{A_j\})\bigr)
   &=\sum_j F_j^\sigma\,\sigma(A_j)\,(F_j^\sigma)^\dagger
   &&\text{if }\sigma\models\mathit{Dist}(\bar q).
\end{aligned}
\]

\caption{Interpretation of quantum predicate formulas in a classical state $\sigma$. Here, the condition $\mathit{Dist}$ is as defined in \eqref{eq:dist}. The Kraus operators $F_j^\sigma=F_j(\sigma(\bar t))_{\sigma(\bar q)}$ satisfy $\sum_jF_j^\sigma(F_j^\sigma)^\dagger\le I$.
Each clause also requires its subformulas to be well-defined;
otherwise $\sigma(A)$ is not well-defined.}
\label{fig:predicate-interpretation}
\end{figure}

\paragraph{Classical--quantum assertions.}
A \emph{cq-assertion} is a pair $(\varphi,A)$, where $\varphi$ is a
classical first-order formula and $A$ is a quantum predicate formula.
For a cq-state $(\sigma,\rho)$ with $\sigma\models\varphi\wedge\mathrm{WD}(A)$, the expectation of $(\varphi,A)$ is
$\Tr(\sigma(A)\rho)$. We define the entailment by
\[
 (\varphi,A)\models(\psi,B)
 \iff
 \varphi\models\psi,\
 \varphi\models\mathrm{WD}(A)\leftrightarrow\mathrm{WD}(B),\ \text{and}\
 \sigma(A)\le\sigma(B)\ \text{for all}\ \sigma\models\varphi\wedge\mathrm{WD}(A).
\]

Let us use $A[e/x]$ to denote the substitution of $e$ for $x$ in $A$. For every classical state $\sigma$, we define $\sigma\bigl(A[e/x]\bigr)=\sigma[x:=\sigma(e)](A)$.

\subsection{Proof System}\label{ssec:rules}

The previous ingredients now fit together in a Hoare triple
$\{\varphi,A\}\,P\,\{\psi,B\}$, with pre- and postconditions being cq-assertions. 
The validity of Hoare triples is then defined as follows, providing formal meanings to the specifications \eqref{eq:rustriple}--\eqref{eq:trottertriple} in Section~\ref{sec:motivating}.

\begin{definition}[validity in \qhl{}]\label{def:source-valid}
A Hoare triple $\{\varphi,A\}P\{\psi,B\}$ is valid for total correctness, denoted as $\models_{\mathit{tot}}\{\varphi,A\}P\{\psi,B\}$,
if for every input cq-state $(\sigma,\rho)$ with
$\sigma\models\varphi\wedge\mathrm{WD}(A)$,
\begin{equation}
\label{eq:validity}
 \Tr\bigl(\sigma(A)\rho\bigr)
     \le
 \sum_{\substack{(\sigma',\rho')\in\llbracket P\rrbracket(\sigma,\rho)\\
                 \sigma'\models\psi\wedge\mathrm{WD}(B)}}
       \Tr\bigl(\sigma'(B)\rho'\bigr).    
\end{equation}
The partial correctness $\models_{\mathit{par}}\{\varphi,A\}P\{\psi,B\}$ is defined similarly, except for adding a $\mathrm{NT}(P)(\sigma,\rho)$ term (accounting for non-termination) to the right-hand side of \eqref{eq:validity}.
\end{definition}

\paragraph{Proof rules.}
Table~\ref{tab:rules} shows the proof rules of \qhl{} for deriving Hoare triples. 
Here, $\mathit{free}(\varphi)$
is the set of free classical variables of the classical formula $\varphi$,
and $\mathit{var}(\bar t)$ is the set of classical variables in $\bar t$. We use $\mathit{cv}(A), \mathit{cv}(\bar q)$
to denote the sets of classical variables in the quantum predicate formula $A$ and quantum variables $\bar q$, respectively. For a program $P$, $\mathit{qv}(P)$ is the set of quantum
variables in $P$, and $\mathit{change}(P)$ consists of classical variables that $P$ can modify. 

Most axioms and rules are self-explanatory.
In (Axiom-Meas), each outcome $j$ sets $x$ to $j$ and contributes
$M_j^\dagger A[j/x]M_j$ to the precondition, and $\forall x.\varphi$ requires
$\varphi$ for every outcome.
(Rule-Mask) moves the masking condition $b$ in the quantum predicate $b\cdot B$ into the classical postcondition.
For example, for a fresh variable $y$, taking $b$ to be $x=y$ allows one to select a single measurement outcome. We can derive the measurement axiom of~\cite{Ying2025} by applying (Axiom-Meas) to
$\{\varphi\vee x\neq y,\,(x=y)\cdot A\}$, followed by (Rule-Mask) and (Rule-Conseq),
provided $\varphi[y/x]\models\forall x.\mathrm{WD}(A)$ (see Appendix~\ref{app:meas}).
In (Rule-Accum1), $F$ contains $k$ Kraus operators and $F'$ contains one,
all of the same quantum type. Both (Rule-Accum1) and (Rule-Accum2) require
$\sum_i F_iF_i^\dagger=I$.
The notation $\propto$ is defined such that for any classical state $\sigma$, $\sigma\models F(\bar t)\propto F'(\bar t)$ holds if for every $i=1,\ldots,k$ and density operator
$\rho$ on this quantum type,
\[
 F_i^\dagger(\sigma(\bar t))\,\rho\,F_i(\sigma(\bar t))
 \;\le\;
 F'^{\dagger}(\sigma(\bar t))\,\rho\,F'(\sigma(\bar t)).
\]
(Rule-Accum1) combines mutually exclusive classical postconditions by
disjunction, while (Rule-Accum2) combines quantum preconditions and
postconditions using a Kraus operator.

\begin{table}[!htbp]
\centering
\footnotesize
\renewcommand{\arraystretch}{1.25}
\begin{tabularx}{\textwidth}{@{}r@{\ \ }Y@{}}
\multicolumn{2}{@{}l@{}}{\textbf{(a)} Partial correctness (\qhlpar{})} \\[2pt]
(\textit{Axiom-Ski}) &
  $\{\varphi,A\}\ \mathbf{skip}\ \{\varphi,A\}$ \qquad\qquad\qquad
(\textit{Axiom-Ass}) 
  $\{\varphi[e/x],A[e/x]\}\ x:=e\ \{\varphi,A\}$ \\[6pt]
(\textit{Axiom-Init}) &
  $\{\varphi,F_B[q](A)\}\ q:=\ket0\ \{\varphi,A\}$ \qquad
(\textit{Axiom-Uni})
  $\{\varphi,F_U(\bar t)[\bar q](A)\}\ U(\bar t)[\bar q]\ \{\varphi,A\}$ \\[6pt]
(\textit{Axiom-Meas}) &
  $\bigl\{\forall x.\,\varphi,\ F_M[\bar q](\{A[j/x]\})\bigr\}\
   x := M[\bar q]\ \{\varphi,\,A\}$ \\[6pt]
(\textit{Rule-Seq}) &
  $\dfrac{\{\varphi, A\}\ P_1\ \{\psi, B\} \quad \{\psi, B\}\ P_2\ \{\theta, C\}}
         {\{\varphi, A\}\ P_1; P_2\ \{\theta, C\}}$ \qquad
(\textit{Rule-Cond}) 
  $\dfrac{\{\varphi \wedge b,\, A\}\ P_1\ \{\psi, B\} \quad \{\varphi \wedge \neg b,\, A\}\ P_0\ \{\psi, B\}}
         {\{\varphi, A\}\ \mathbf{if}\ b\ \mathbf{then}\ P_1\ \mathbf{else}\ P_0\ \{\psi, B\}}$ \\[6pt]
(\textit{Rule-Loop-par}) &
  $\dfrac{\{\varphi \wedge b,\, A\}\ P\ \{\varphi, A\}}
         {\{\varphi, A\}\ \mathbf{while}\ b\ \mathbf{do}\ P\ \{\varphi \wedge \neg b,\, A\}}$ \qquad
(\textit{Rule-Mask})
  $\dfrac{\{\varphi, A\}\ P\ \{\psi,\ b\cdot B\}}
         {\{\varphi, A\}\ P\ \{\psi\wedge b,\ B\}}$ \\[6pt]
(\textit{Rule-Conseq}) &
  $\dfrac{(\varphi', A') \models (\varphi, A) \quad \{\varphi, A\}\ P\ \{\psi, B\} \quad (\psi, B) \models (\psi', B')}
         {\{\varphi', A'\}\ P\ \{\psi', B'\}}$ \\[6pt]
(\textit{Rule-Accum1}) &
  $\dfrac{\begin{gathered}
    \{\varphi,A_i\}\ P\ \{\psi_i,B\}\quad\text{for every }i,\qquad
    (\forall i_1,i_2)\bigl(i_1\ne i_2\rightarrow\neg(\psi_{i_1}\wedge\psi_{i_2})\bigr)\\
      \quad F(\bar t)\propto F'(\bar t)\quad
    \bar q\cap\mathit{qv}(P)=\emptyset\quad
    [\mathit{var}(\bar t)\cup\mathit{cv}(\bar q)]\cap\mathit{change}(P)=\emptyset
    \end{gathered}}
    {\{\varphi,F(\bar t)[\bar q](\{A_i\})\}\ P\
     \{\bigvee_i\psi_i,F'(\bar t)[\bar q](B)\}}$ \\[6pt]
(\textit{Rule-Accum2}) &
  $\dfrac{\begin{gathered}
    \{\varphi,A_i\}\ P\ \{\psi,B_i\}\quad\text{for every }i\qquad
    \psi\rightarrow\textstyle\bigwedge_i\mathrm{WD}(B_i)\\
    \bar q\cap\mathit{qv}(P)=\emptyset\quad
    [\mathit{var}(\bar t)\cup\mathit{cv}(\bar q)]\cap\mathit{change}(P)=\emptyset
    \end{gathered}}
    {\{\varphi,F(\bar t)[\bar q](\{A_i\})\}\ P\
     \{\psi,F(\bar t)[\bar q](\{B_i\})\}}$ \\[6pt]
\multicolumn{2}{@{}l@{}}{\textbf{(b)} Total correctness (\qhltot{}): replace (Rule-Loop-par) by} \\
(\textit{Rule-Loop-tot}) &
  $\dfrac{\{\varphi \wedge b,\, A\}\ P\ \{\varphi, A\} \qquad
          \{\varphi \wedge b \wedge t = z,\, I\}\ P\ \{t < z,\, I\} \qquad
          \varphi \rightarrow t \geq 0}
         {\{\varphi, A\}\ \mathbf{while}\ b\ \mathbf{do}\ P\ \{\varphi \wedge \neg b,\, A\}}$ \\[6pt]
\end{tabularx}
\caption{The QHL$^+$ proof rules~\cite{Ying2025}, with a variant of (Axiom-Meas)
and the added (Rule-Mask).
Here, $F_B$, $F_U$, $F_M$, and the mask $b\cdot B$ are the instances of $F(\bar t)$
defined above.
In (Rule-Accum1) and (Rule-Accum2), $i$ indexes the Kraus operators of $F$.
In (Rule-Loop-tot), the bound function $t$ is an integer expression and $z$ is a fresh integer variable, as usual for loops~\cite{AptOlderog2009}.
}
\label{tab:rules}
\end{table}

\section{Formalization in Lean and Library Organization}\label{sec:faithful}\label{sec:faith}\label{sec:arch}

We now describe how QEDIA handles the classical--quantum interplay
consistently at the three levels of Section~\ref{sec:motivating}. The
library follows these levels, each building on those below it.

This consistency comes from the way \qhl{} combines classical and quantum
components. Classical components keep their roles from classical programs,
and quantum components are built around them. The formalization follows
this combination at every level. The classical components are expressed
directly in Lean, and the library builds the quantum components on them.
At \textbf{(L0)}, quantum variables are names from a finite set, and each
set of names determines the Hilbert space of those variables. Classical
data, such as array indices, can therefore select the quantum variables,
and hence the space, on which an operation acts. At \textbf{(L1)},
classical expressions and guards are Lean functions on classical states.
Classical expressions also supply the indices of subscripted quantum
variables and the parameters of gates. Proving the condition
$\mathit{Dist}(\bar q)$ of~\eqref{eq:dist} thus reduces to classical
reasoning about indices. At \textbf{(L2)}, the classical formula of a
cq-assertion is a Lean predicate on classical states, and its quantum
predicate is parameterized by classical variables. Classical reasoning
about specifications therefore takes place directly in Lean's logic.

Table~\ref{tab:lean-map} lists the Lean counterparts of the notions of
Section~\ref{sec:logic}, with their types.
Sections~\ref{ssec:lean-programs} and~\ref{ssec:lean-semantics} formalize
\textbf{(L0)} and \textbf{(L1)}: state spaces, programs, and their
execution paths. Sections~\ref{ssec:lean-assertions}
and~\ref{ssec:lean-rules} formalize \textbf{(L2)}: assertions and
validity, and the proof systems with their soundness.
Section~\ref{sec:qpv-metrics} reports the size of the library and the
examples.

\begin{table}[t]
\caption{Lean counterparts of the notions of Section~\ref{sec:logic}.
Types are simplified and parameters are omitted; for a Lean type, the third column gives the type it stands for. $\mathcal L(\mathcal H)$ denotes the linear operators on $\mathcal H$, and values of \texttt{Assn} are quantum predicates.}
\label{tab:lean-map}
\footnotesize
\begin{tabular}{@{}l l >{\raggedright\arraybackslash}p{0.39\linewidth}@{}}
\toprule
Section~\ref{sec:logic} & Lean & Type \\
\midrule
classical-state type $\Sigma$ & \texttt{$\sigma$T} & \texttt{Type} \\
space $\mathcal H_S$ of $S\subseteq L$ & \texttt{HS D S} & $\mathbb C^{\prod_{q\in S}D_q}$ \\
expression $e$ of type $T$, guard $b$ & Lean functions & $\Sigma\to T$, $\Sigma\to\mathsf{Bool}$ \\
sequence $\bar q=(q_1,\ldots,q_k)$ & \texttt{RegExpr} & $\Sigma\to\{0,\ldots,k-1\}\to L$ \\
program $P$ & \texttt{Prog} & inductive type for the syntax of Figure~\ref{fig:prog-syntax} \\
quantum predicate formula $A$ & \texttt{Assn} & $\Sigma\to\mathcal L(\mathcal H_{\mathit{all}})$ \\
cq-assertion $(\varphi,A)$ & \texttt{SourceAssertion} & $(\Sigma\to\texttt{Prop})\times\texttt{Assn}\times(\Sigma\to\texttt{Prop})$ \\
masking by $\varphi\wedge\mathrm{WD}(A)$ & \texttt{toAssn} & $\texttt{SourceAssertion}\to\texttt{Assn}$ \\
$\models_{\mathit{par}},\ \models_{\mathit{tot}}$ & \texttt{validPar}, \texttt{validTot} & $\texttt{Assn}\to\texttt{Prog}\to\texttt{Assn}\to\texttt{Prop}$ \\
$\vdash_{\mathit{par}},\ \vdash_{\mathit{tot}}$ & \texttt{SourceDerivPar}, \texttt{SourceDerivTot} & $\texttt{SourceAssertion}\to\texttt{Prog}\to\texttt{SourceAssertion}\to\texttt{Prop}$ \\
\bottomrule
\end{tabular}
\end{table}

\subsection{State Spaces and Programs}\label{ssec:lean-programs}

At \textbf{(L0)}, we use, on top of Mathlib, the labeled Hilbert spaces
and Dirac notation of~\cite{Zhou2023}. Each quantum variable $q$ in a
finite set $L$ of variable names has a finite, nonempty set $D_q$ of basis
labels. A set $S\subseteq L$ has the Hilbert space
\[
 \mathcal H_S=\mathbb C^{\prod_{q\in S}D_q}\cong\bigotimes_{q\in S}\mathbb C^{D_q},
\]
with $\mathcal H_{\mathit{all}}=\mathcal H_L$. An operator $O_S$ on
$\mathcal H_S$ acts on $\mathcal H_{\mathit{all}}$ through its
\emph{cylindrical extension} $O_S\otimes I$, where $I$ is the identity on
the variables in $L\setminus S$. Thus, whichever quantum variables the
classical state selects, every quantum operation becomes an operator on
$\mathcal H_{\mathit{all}}$, on which quantum states and quantum
predicates are also defined.

At \textbf{(L1)}, programs are deeply embedded~\cite{Boulton1992Embedding},
since the operational semantics and the proof rules are defined by cases
on program syntax. Programs form the inductive type \texttt{Prog} for the
syntax of Figure~\ref{fig:prog-syntax}. For simplicity, classical
expressions are shallowly embedded. Lean represents a classical
expression $e$ by its evaluation $\sigma\mapsto\sigma(e)$, a function on
the type $\Sigma$ of classical states. A guard $b$ thus has type
$\Sigma\to\mathsf{Bool}$, and proofs about classical computations use
Lean's arithmetic.

Classical expressions also select quantum variables. A sequence
$\bar q=(q_1,\ldots,q_k)$ may contain subscripted quantum variables such
as $q[i]$, so the quantum variables $\sigma(\bar q)$ it refers to depend
on the classical state. We represent $\bar q$ by a function
$r:\Sigma\to\{0,\ldots,k-1\}\to L$ (\texttt{RegExpr} in
Table~\ref{tab:lean-map}) with $r(\sigma)=\sigma(\bar q)$. Then
$\sigma\models\mathit{Dist}(\bar q)$ holds if and only if $r(\sigma)$ is
injective. Lean represents a quantum array by an injective map from its
indices into $L$. The condition $\mathit{Dist}(\bar q)$ therefore reduces
to conditions on index values, which Lean proves with its arithmetic.

The \texttt{Prog} constructors for assignments, initializations, gates,
and measurements take a classical update, a quantum operation, or both.
The constructor for $x:=e$ takes the update
$\sigma\mapsto\sigma[x:=\sigma(e)]$. The constructor for initialization
takes a finite set of quantum variables, all reset to $\ket0$ at once, as
in $\bar q:=\ket0$. The constructor for $U(\bar t)[\bar q]$
takes $r$ and, in each classical state $\sigma$, the operator
$U(\sigma(\bar t))$ on
$\mathcal H_{r(\sigma)(0)}\otimes\cdots\otimes\mathcal H_{r(\sigma)(k-1)}$.
This space is identified with $\mathcal H_{\sigma(\bar q)}$ when
$\mathit{Dist}(\bar q)$ holds. The constructor for $x:=M[\bar q]$ takes
both: $r$ and the measurement operators $M_j^\sigma$ in each $\sigma$, and
an update $u_j:\Sigma\to\Sigma$ for each outcome $j$, such as
$\sigma\mapsto\sigma[x:=j]$. Since these operators depend on $\sigma$,
Lean checks, when a program is written, that for every classical state
each gate is unitary and each measurement satisfies the completeness
relation.

\subsection{Execution Paths}\label{ssec:lean-semantics}

Lean represents the denotational semantics
$\llbracket P\rrbracket(\sigma,\rho)$ of Section~\ref{sec:logic} by
execution paths, on which the soundness proofs of
Section~\ref{ssec:lean-rules} rely. An
\emph{execution path} of a configuration $c=(P,\sigma,\rho)$ is a finite
sequence of transitions of Figure~\ref{fig:operational-rules} starting
from $c$. Let $\mathcal T_n(c)$ be the finite set of execution paths of
length $n$ that end in a terminal configuration. We write
$(\downarrow,\sigma_h,\rho_h)$ for the terminal configuration of a path
$h$. Each such path contributes $(\sigma_h,\rho_h)$ to
$\llbracket P\rrbracket(\sigma,\rho)$. Sums over this multiset, such as
the one in~\eqref{eq:validity}, thus become sums over paths: for every nonnegative function $f$ on cq-states,
\begin{equation}\label{eq:lean-paths}
 \sum_{(\sigma',\rho')\in\llbracket P\rrbracket(\sigma,\rho)}f(\sigma',\rho')
 =\sum_{n=0}^{\infty}\sum_{h\in\mathcal T_n(c)}f(\sigma_h,\rho_h).
\end{equation}
For $f(\sigma',\rho')=\Tr(\rho')$, the partial sums of the right-hand side
of~\eqref{eq:lean-paths} are bounded by $\Tr(\rho)$. The series therefore
converges. Partial correctness uses the resulting non-termination
probability
$\mathrm{NT}(P)(\sigma,\rho)=\Tr(\rho)-\sum_n\sum_{h\in\mathcal T_n(c)}\Tr(\rho_h)$.

\subsection{Assertions and Validity}\label{ssec:lean-assertions}

At \textbf{(L2)}, Lean keeps the classical and quantum parts of a
cq-assertion apart. Like classical expressions, quantum predicate formulas
are shallowly embedded. Lean represents a formula $A$ by its
interpretation $\sigma\mapsto\sigma(A)$
(Figure~\ref{fig:predicate-interpretation}), a function of type
\texttt{Assn}. As a total function, the interpretation does not record
where $A$ is well-defined, so the well-definedness condition
$\mathrm{WD}(A)$ is kept as a separate component. A cq-assertion
$(\varphi,A)$ is thus represented by a \texttt{SourceAssertion} with three
components: $\varphi$ and $\mathrm{WD}(A)$ as predicates on $\Sigma$, and
the interpretation of $A$. The entailment $(\varphi,A)\models(\psi,B)$ of
Section~\ref{sec:logic} uses these components separately. Its first two
conditions, $\varphi\models\psi$ and
$\varphi\models\mathrm{WD}(A)\leftrightarrow\mathrm{WD}(B)$, become
implications between Lean predicates. The third, $\sigma(A)\le\sigma(B)$,
compares the interpretations in the states satisfying
$\varphi\wedge\mathrm{WD}(A)$.

Definition~\ref{def:source-valid} uses a cq-assertion $(\varphi,A)$ only
through the expectation $\Tr(\sigma(A)\rho)$ for classical states
$\sigma\models\varphi\wedge\mathrm{WD}(A)$. The function \texttt{toAssn}
therefore represents $(\varphi,A)$ by a single quantum predicate, which is
$\sigma(A)$ for these states and $0$ for the others. This is the masking
$b\cdot A$ of Section~\ref{sec:logic}, with the formula
$\varphi\wedge\mathrm{WD}(A)$ in place of $b$.

With \texttt{toAssn}, validity is stated for single quantum predicates
$\mathcal A$ and $\mathcal B$ of type \texttt{Assn}. The proposition
$\mathtt{validTot}\ \mathcal A\ P\ \mathcal B$ states that, for every
classical state $\sigma$ and positive operator $\rho$,
\begin{equation}\label{eq:lean-valid}
 \Tr\bigl(\mathcal A(\sigma)\rho\bigr)\le
 \sum_{n=0}^{\infty}\sum_{h\in\mathcal T_n(P,\sigma,\rho)}
   \Tr\bigl(\mathcal B(\sigma_h)\rho_h\bigr).
\end{equation}
The proposition \texttt{validPar} adds $\mathrm{NT}(P)(\sigma,\rho)$ to the
right-hand side. For $\mathcal A=\texttt{toAssn}(\varphi,A)$ and
$\mathcal B=\texttt{toAssn}(\psi,B)$, \eqref{eq:lean-valid} is equivalent
to the validity of Definition~\ref{def:source-valid}. By~\eqref{eq:lean-paths},
its right-hand side is a sum over $\llbracket P\rrbracket(\sigma,\rho)$,
as in~\eqref{eq:validity}. In this sum, masking makes every term with
$\sigma_h\not\models\psi\wedge\mathrm{WD}(B)$ zero. This matches the
condition on the sum in~\eqref{eq:validity}. On the left-hand side,
masking gives zero when $\sigma\not\models\varphi\wedge\mathrm{WD}(A)$, so
such states impose no condition. Finally, both sides
of~\eqref{eq:lean-valid} are linear in $\rho$. Quantifying over positive
operators is therefore equivalent to quantifying over the partial density
operators of Definition~\ref{def:source-valid}.

\subsection{Proof Systems and Soundness}\label{ssec:lean-rules}

The inductive propositions \texttt{SourceDerivPar} and
\texttt{SourceDerivTot} have one constructor for each axiom and rule of
Table~\ref{tab:rules}, and the constructor for (Rule-Conseq) uses the
entailment of Section~\ref{ssec:lean-assertions}.

For the axioms, Lean computes the precondition from the postcondition
$(\varphi,A)$. As for the constructors of Section~\ref{ssec:lean-programs},
classical updates and quantum operations are treated separately. A
classical update is handled by evaluation. For $x:=e$,
the precondition evaluates $\varphi$, $\mathrm{WD}(A)$, and the
interpretation of $A$ in $\sigma[x:=\sigma(e)]$. This is the substitution
$[e/x]$ of (Axiom-Ass), as defined in Section~\ref{sec:logic}. A quantum
operation is handled by its transformation. For initialization and gates,
the precondition applies $F_B$ or $F_U$ in each classical state. A
measurement combines both. In a classical state
$\sigma\models\mathit{Dist}(\bar q)$, the quantum predicate of its
precondition is interpreted as
\[
 \sum_j\bigl(M_j^\sigma\bigr)^\dagger\,
   \bigl(u_j(\sigma)\bigr)(A)\,M_j^\sigma.
\]
This is $F_M$ of (Axiom-Meas), with $\sigma[x:=j]$ generalized to
$u_j(\sigma)$. As in (Axiom-Meas), $\varphi$ and $\mathrm{WD}(A)$ must
hold after every update $u_j$. Since $A$ is interpreted after the update,
it can depend on the outcome, as the invariant $\inv(m)$ of
Example~\ref{ssec:rus} does.

Since classical expressions and quantum predicate formulas are shallowly
embedded, side conditions about their variables are stated semantically.
In
(Rule-Loop-tot), the fresh variable $z$ becomes an integer quantified
outside the program, so the loop body cannot modify it. In (Rule-Accum1)
and (Rule-Accum2), the conditions on $\mathit{change}(P)$ and
$\mathit{qv}(P)$ become the following: every transition of $P$ preserves
the values of $\bar t$ and $\sigma(\bar q)$, and $P$ acts on no quantum
variable in $\sigma(\bar q)$. These systems satisfy the following
soundness theorem.
\begin{theorem}[soundness]\label{thm:sound}
For all cq-assertions $(\varphi,A)$, $(\psi,B)$ and programs $P$,
\[
\begin{aligned}
 \mathtt{SourceDerivPar}\ (\varphi,A)\ P\ (\psi,B)
   &\Longrightarrow\mathtt{validPar}\ (\texttt{toAssn}(\varphi,A))\ P\ (\texttt{toAssn}(\psi,B)),\\
 \mathtt{SourceDerivTot}\ (\varphi,A)\ P\ (\psi,B)
   &\Longrightarrow\mathtt{validTot}\ (\texttt{toAssn}(\varphi,A))\ P\ (\texttt{toAssn}(\psi,B)).
\end{aligned}
\]
\end{theorem}
The proof is by induction on the derivation, with one semantic lemma per
constructor. For the axioms, the lemmas follow directly from the
preconditions above. For (Rule-Mask), the lemma is immediate, since
masking twice combines the conditions:
$\texttt{toAssn}(\psi,b\cdot B)=\texttt{toAssn}(\psi\wedge b,B)$. For
sequential composition, each execution path of $P_1;P_2$ splits into an
execution path of $P_1$ followed by one of $P_2$. The lemma thus combines
the validity of the two premises through the intermediate assertion. For
loops, induction on the length of execution paths gives the
partial-correctness bound. Induction on the initial value of the bound
function $t$ gives the total-correctness bound. The soundness theorem uses
only Lean's standard axioms \texttt{propext}, \texttt{Classical.choice},
and \texttt{Quot.sound}.

\subsection{Code Metrics}\label{sec:qpv-metrics}

Tables~\ref{tab:metrics} and~\ref{tab:metrics-ex} report the size of the
library and example files, including integration tests.
Table~\ref{tab:metrics} groups the library components by level.
Excluding blank lines, the Spec column counts lines outside proofs and
comments, and the Proof column counts lines strictly inside proof bodies.

At the declaration level, the library contains 2{,}475
\texttt{theorem}/\texttt{lemma} declarations and 673 definitions. The
examples contain 1{,}117 such declarations and 575 definitions.
\begingroup
\small
\setlength{\LTcapwidth}{\textwidth}
\begin{longtable}{@{}l r r@{}}
\caption{Code metrics of the library by level and component: upper
bounds on line counts.}\label{tab:metrics} \\
\toprule
Component & Spec & Proof \\
\midrule
\endfirsthead
\multicolumn{3}{c}{\small\itshape Table~\thetable{} (continued)} \\
\toprule
Component & Spec & Proof \\
\midrule
\endhead
\bottomrule
\endfoot
\bottomrule
\endlastfoot
\multicolumn{3}{@{}l}{\textbf{(L0)} Mathematics and quantum mechanics} \\
\quad Foundations & $<$250 & $<$200 \\
\quad Hilbert spaces and operators & $<$1{,}100 & $<$1{,}100 \\
\quad Multi-qubit states and operators & $<$2{,}500 & $<$3{,}800 \\
\quad Gates & $<$900 & $<$650 \\
\quad Dirac notation & $<$950 & $<$600 \\
\quad Superoperators & $<$450 & $<$400 \\
\multicolumn{3}{@{}l}{\textbf{(L1)} Program syntax and semantics} \\
\quad Programs & $<$500 & $<$250 \\
\quad Semantics & $<$700 & $<$750 \\
\multicolumn{3}{@{}l}{\textbf{(L2)} Specification and proof system} \\
\quad Logic & $<$5{,}600 & $<$5{,}000 \\
\midrule
Tactics & $<$100 & $<$50 \\
\midrule
Total & $<$13{,}000 & $<$13{,}000 \\
\end{longtable}
\endgroup
\begin{table}[!htbp]
\centering
\small
\caption{Code metrics of the program examples: upper bounds on line counts.}
\label{tab:metrics-ex}
\begin{tabular}{@{}l r r@{}}
\toprule
Example & Spec & Proof \\
\midrule
Bell-state preparation & $<$100 & $<$50 \\
Parallel Hadamard & $<$150 & $<$300 \\
Quantum Fourier transform & $<$750 & $<$1{,}150 \\
Inverse QFT & $<$350 & $<$600 \\
Quantum phase estimation & $<$500 & $<$800 \\
Cluster-state preparation & $<$100 & $<$100 \\
Repeat-until-success & $<$150 & $<$100 \\
Trotter product formula (TFIM) & $<$550 & $<$550 \\
Bravyi--Kitaev transformation & $<$350 & $<$400 \\
Fermi--Hubbard LCU block encoding & $<$1{,}150 & $<$1{,}700 \\
Fermi--Hubbard QSP & $<$1{,}450 & $<$1{,}900 \\
Other examples and integration tests & $<$1{,}100 & $<$950 \\
\midrule
Total & $<$6{,}500 & $<$8{,}500 \\
\bottomrule
\end{tabular}
\end{table}

\section{The Prover Agent: Building the Library and Verifying Programs}\label{sec:qpv}
We construct the library of Section~\ref{sec:faithful} and the program
proofs of Section~\ref{sec:suite} with a prover agent. It consists of a
general-purpose core agent (Section~\ref{sec:qpv-design}) and a quantum
extension (Section~\ref{sec:qpv-quantum}). The agent first builds the
library, formalizing \qhl{} and the mathematics it relies on, and then
uses the library to verify programs. Throughout,
the workflow separates fixing what is proved from proving it. A leader
fixes the statements, subagents prove them, and the leader accepts a proof
only if it proves the fixed statements. Figure~\ref{fig:qpv} summarizes
the workflow.

\subsection{Core Agent}
\label{sec:qpv-design}\label{sec:qpv-orch}
The core agent follows the approach of recent LLM-based theorem-proving
agents~\cite{li2026merleanproverrecursiveloopingharness,pmlr-v306-liu26gp,breen2026axproverdeepreasoningagentic}.
Given a mathematical statement or source, it produces Lean declarations
in three stages: task definition, proof construction, and acceptance. The
leader coordinates proof-writing subagents, and both are LLM agents. Human
reviewers supervise the process and review the statements, the library,
and the proofs. All agents use Mathlib and the library. They
follow shared
guidelines: preserve the target statement and its assumptions, search for
existing results before proving new ones, and place reusable lemmas in
the library.

\paragraph{Task definition.}
The leader reads the source and the library and writes a task
specification: the target statements, their permitted assumptions, the
required declarations, and the permitted changes. It reviews whether the
specification expresses the intended mathematics and gives a meaningful
guarantee. It then assigns the task to a subagent once its prerequisites
are completed.

\paragraph{Proof construction.}
A subagent constructs each proof with Lean tools. These tools report proof
goals and diagnostics, build files, and search Mathlib and the library.
The subagent first tries existing lemmas and automation, and then edits
and repairs the proof with Lean's feedback. It also checks axioms and
inspects its own source locally, and reviews the proof against the
guidelines, so failures can be repaired before submission.

\paragraph{Acceptance.}
The leader runs five checks on the combined changes of all subagents,
returns failed changes for repair, and repeats the affected checks after
any later change.
Checks~\ref{check:faithful} and~\ref{check:scope} relate the submitted
declarations to the task specification. Checks 2--4 ensure that the proofs
are complete and checked by Lean's kernel with only the standard axioms.
\begin{enumerate}[leftmargin=1.6em, itemsep=3pt, topsep=3pt]
\item\label{check:faithful} \textbf{Faithful statements (Lean, LLM review).}
  The leader reviews whether the required declarations, and the
  definitions they depend on, express the mathematics fixed in the task
  specification. When a task verifies a program, Lean also checks that its
  validity theorem has the type of the target Hoare triple fixed there.
\item \textbf{Clean build (Lean, Python).}
  The library, including all submitted modules, builds without errors
  under the required settings.
\item\label{check:axioms} \textbf{Axiom dependencies (Lean, Python).}
  For each required theorem, Lean reports the axioms it uses. A script
  checks that they are exactly \texttt{propext},
  \texttt{Classical.choice}, and \texttt{Quot.sound}.
\item \textbf{Kernel-checked proofs (Lean, LLM review).}
  The leader inspects the submitted source for \texttt{sorry},
  \texttt{admit}, \texttt{native\_decide}, and options that bypass kernel
  checking. Lean checks the proofs under the required settings.
\item\label{check:scope} \textbf{Protected statements and definitions (LLM review).}
  The leader compares the submitted changes with the approved baseline.
  It reviews whether they stay within the permitted changes and preserve
  the target statements and protected definitions. When a task extends
  the library, the leader also reviews
  new or changed definitions against their mathematical sources.
\end{enumerate}

\subsection{Quantum Extension}
\label{sec:qpv-quantum}
The quantum extension supports applying \qhl{} rules and adds
quantum-specific checks. In program verification, a task consists of a
quantum program and a target Hoare triple. Its required declarations are
the program as a \texttt{Prog} term, the precondition and postcondition, a
\texttt{SourceDerivTot} or \texttt{SourceDerivPar} derivation, and the
validity theorem. The agents construct the derivation, and the leader
obtains the validity theorem from it by Theorem~\ref{thm:sound}. Program proofs thus reason about assertions
rather than execution paths (Section~\ref{ssec:lean-semantics}).

\paragraph{Verification conditions.}
Applying a rule turns its premises into Lean proof goals. For the axioms,
Lean computes the precondition from the postcondition
(Section~\ref{ssec:lean-rules}). For the other rules, the agents choose
what the conclusion leaves open, such as intermediate assertions, loop
invariants, and bound functions. The resulting verification conditions
split into classical and quantum conditions, as in
Section~\ref{sec:faithful}. Classical conditions include the classical
parts of entailments, $\mathit{Dist}(\bar q)$, and the nonnegativity of a
bound function. They are goals about Lean predicates on classical states,
proved with Lean's logic and arithmetic. Quantum conditions include
$\sigma(A)\le\sigma(B)$ in the L\"owner order and the conditions
$F\propto F'$ and $\sum_iF_iF_i^\dagger=I$ of the accumulation rules. They
are proved with library results about operators and quantum variables.

\paragraph{Quantum-specific checks.}
When a task verifies a program, the review in check~\ref{check:faithful}
also examines whether the program implements the intended algorithm, with
each gate defined by an explicit unitary in the library and each reused
component a verified subprogram (Section~\ref{ssec:ex-overview}). For each
example, Lean verifies that the precondition has positive expectation on
some input cq-state. This
excludes preconditions with zero expectation on every input, for which the
left-hand side of~\eqref{eq:validity} vanishes and the Hoare triple holds
for every program. For program tasks, the statements and definitions protected by
check~\ref{check:scope} include the program, the target Hoare triple, the
proof systems, and the soundness proofs.

\definecolor{qpvink}{HTML}{253247}
\definecolor{qpvmuted}{HTML}{526174}
\definecolor{qpvartifill}{HTML}{FAF6E8}
\definecolor{qpvagentblue}{HTML}{0072B2}
\definecolor{qpvfailred}{HTML}{B84D19}
\definecolor{qpvpassgreen}{HTML}{007C60}
\newcommand{\qpvlean}{\mbox{\textsf{L}\kern-.12em$\exists\forall$\kern-.12em\textsf{N}}}
\tikzset{
  qpvbase/.style={align=center, font=\small, text=qpvink,
    rounded corners=2pt, inner xsep=1.2mm, inner ysep=1mm, line width=.7pt},
  qpvagent/.style={qpvbase, draw=qpvagentblue, fill=qpvagentblue!6},
  qpvtool/.style={qpvbase, draw=qpvink!65, fill=white},
  qpvartifact/.style={qpvbase, draw=qpvink!65, fill=qpvartifill},
  qpvcheck/.style={qpvtool, font=\footnotesize, inner xsep=1mm, inner ysep=.8mm},
  qpvflow/.style={-{Latex[length=2mm]}, draw=qpvink, line width=.85pt},
  qpvrepair/.style={qpvflow, draw=qpvfailred},
  qpvlabel/.style={font=\scriptsize, fill=white, inner sep=1pt,
    text=qpvmuted, align=center},
}
\begin{figure*}[htp]
\centering
\resizebox{\textwidth}{!}{%
\begin{tikzpicture}[x=1mm,y=1mm]
\node[qpvartifact, text width=163mm, minimum height=12mm] (qpvguides) at (86,0)
  {\faListOl\quad\textbf{Shared proof-writing guidelines}\\[-1pt]
   {\footnotesize Preserve the target; reuse existing results; place reusable lemmas in the library.}};

\begin{scope}[on background layer]
  \filldraw[fill=qpvagentblue!2, draw=qpvink!25, rounded corners=3pt]
    (4,-17) rectangle (84,-148);
  \filldraw[fill=qpvagentblue!2, draw=qpvink!25, rounded corners=3pt]
    (96,-17) rectangle (172,-101);
\end{scope}
\node[font=\small\bfseries,text=qpvink] at (44,-21) {Leader};
\node[font=\small\bfseries,text=qpvink] at (134,-21) {Subagents};
\node[qpvagent,text width=65mm,minimum height=18mm] (qpvleader) at (44,-35)
  {\faRobot\quad\textbf{Leader: define and assign tasks}\\[1pt]
   {\footnotesize Input: statement or mathematical source\\
    Write and review task specifications.}};
\node[qpvagent,text width=65mm,minimum height=18mm] (qpvworker) at (134,-35)
  {\faRobot\quad\textbf{Proof-writing subagents}\\[1pt]
   {\footnotesize Construct and repair proofs;\\
    review them against the guidelines.}};
\draw[qpvflow] (qpvleader.east) --
  node[qpvlabel,above=1mm]{tasks} (qpvworker.west);
\coordinate (qpvleadguide) at ($(qpvleader.north west)+(7mm,0)$);
\draw[qpvflow,draw=qpvmuted,densely dotted] (qpvguides.south -| qpvleadguide) -- (qpvleadguide);
\coordinate (qpvworkguide) at ($(qpvworker.north west)+(7mm,0)$);
\draw[qpvflow,draw=qpvmuted,densely dotted] (qpvguides.south -| qpvworkguide) -- (qpvworkguide);

\node[qpvtool,text width=65mm,minimum height=18mm] (qpvtools) at (134,-62)
  {\qpvlean\quad\faPython\quad\textbf{Lean tools and local checks}\\[1pt]
   {\footnotesize Goals, diagnostics, builds, lemma search;\\
    tactics, local axiom and source checks.}};
\draw[qpvflow,{Latex[length=2mm]}-{Latex[length=2mm]}]
  (qpvworker.south) -- node[qpvlabel,right=1mm]{edit / check / repair} (qpvtools.north);
\node[qpvtool,text width=65mm,minimum height=15mm] (qpvsubmit) at (134,-89)
  {\textbf{Submit proofs}\\[1pt]
   {\footnotesize The leader decides acceptance.}};
\draw[qpvflow] (qpvtools.south) -- (qpvsubmit.north);

\node[qpvcheck,text width=65mm,minimum height=20mm] (qpvcontent) at (44,-68)
  {\textbf{1. Faithful statements}\\[1pt]
   \faRobot\; Leader reviews the intended mathematics.\\
   \qpvlean\; Check validity-theorem types.};
\node[qpvcheck,text width=30mm,minimum height=16mm] (qpvbuild) at (25.5,-90)
  {\textbf{2. Clean build}\\[1pt]
   \qpvlean\;\faPython\\
   Integrated library; required settings};
\node[qpvcheck,text width=30mm,minimum height=16mm] (qpvaxioms) at (62.5,-90)
  {\textbf{3. Axiom dependencies}\\[1pt]
   \qpvlean\;\faPython\\
   Required theorems};
\node[qpvcheck,text width=30mm,minimum height=16mm] (qpvsource) at (25.5,-110)
  {\textbf{4. Kernel-checked}\\[-1pt]\textbf{proofs}\\[1pt]
   \qpvlean\;\faRobot\; Source checks\\
   No \texttt{sorry}, \texttt{admit},\\[-1pt]
   or \texttt{native\_decide}};
\node[qpvcheck,text width=30mm,minimum height=16mm] (qpvscope) at (62.5,-110)
  {\textbf{5. Protected statements}\\[-1pt]\textbf{and definitions}\\[1pt]
   \faRobot\\
   Baseline comparison\\
   Permitted changes};
\begin{scope}[on background layer]
  \node[draw=qpvink,fill=white,dashed,rounded corners=3pt,
    line width=1pt,inner xsep=1.5mm,inner ysep=3mm,
    fit=(qpvcontent)(qpvbuild)(qpvaxioms)(qpvsource)(qpvscope)] (qpvgate) {};
\end{scope}
\node[font=\small\bfseries,fill=white,text=qpvink,inner sep=1mm]
  at (qpvgate.north) {Acceptance checks};
\draw[qpvflow] (qpvsubmit.west) -- (qpvgate.east |- qpvsubmit.west);
\node[qpvtool,text width=65mm,minimum height=12mm,
  draw=qpvpassgreen,fill=qpvpassgreen!8] (qpvaccept) at (44,-137)
  {\textbf{Accept the combined changes}\\[1pt]
   {\footnotesize All five checks pass.}};
\draw[qpvflow,draw=qpvpassgreen] (qpvgate.south) -- (qpvaccept.north);
\draw[qpvrepair] (qpvgate.west) -- (0,0 |- qpvgate.west) -- (0,-35) -- (qpvleader.west);
\node[qpvlabel,rotate=90,text=qpvfailred] at (0,-62) {failure: repair};

\node[qpvartifact,text width=65mm,minimum height=32mm,
  dashed,line width=.9pt] (qpvquantum) at (134,-131)
  {\faPuzzlePiece\quad\textbf{Quantum extension}\\[2pt]
   {\footnotesize Apply \qhl{} rules\\[2pt]
    Quantum-specific checks}\\[2pt]
   {\scriptsize Used to build the library and to verify programs.}};
\draw[qpvflow,densely dotted,draw=qpvmuted]
  (qpvquantum.east) -- (177,-131) -- (177,-35) -- (qpvworker.east);
\draw[white,line width=2.8pt] (90,-92) -- (90,-86);
\draw[qpvflow,densely dotted,draw=qpvmuted]
  (qpvquantum.west) -- (90,-131) -- (90,-39)
  -- ($(qpvleader.east)+(0,-4mm)$);
\end{tikzpicture}}
\vspace{1mm}

{\footnotesize
\faRobot~LLM review\qquad
\qpvlean~Lean~4 checking\qquad
\faPython~Python scripts}
\caption{The workflow of the core agent with the quantum extension. Icons
show whether Lean, Python scripts, or LLM review performs each acceptance
check of Section~\ref{sec:qpv-orch}. Dotted arrows show where the agents
use the guidelines and the quantum extension.}
\label{fig:qpv}
\end{figure*}
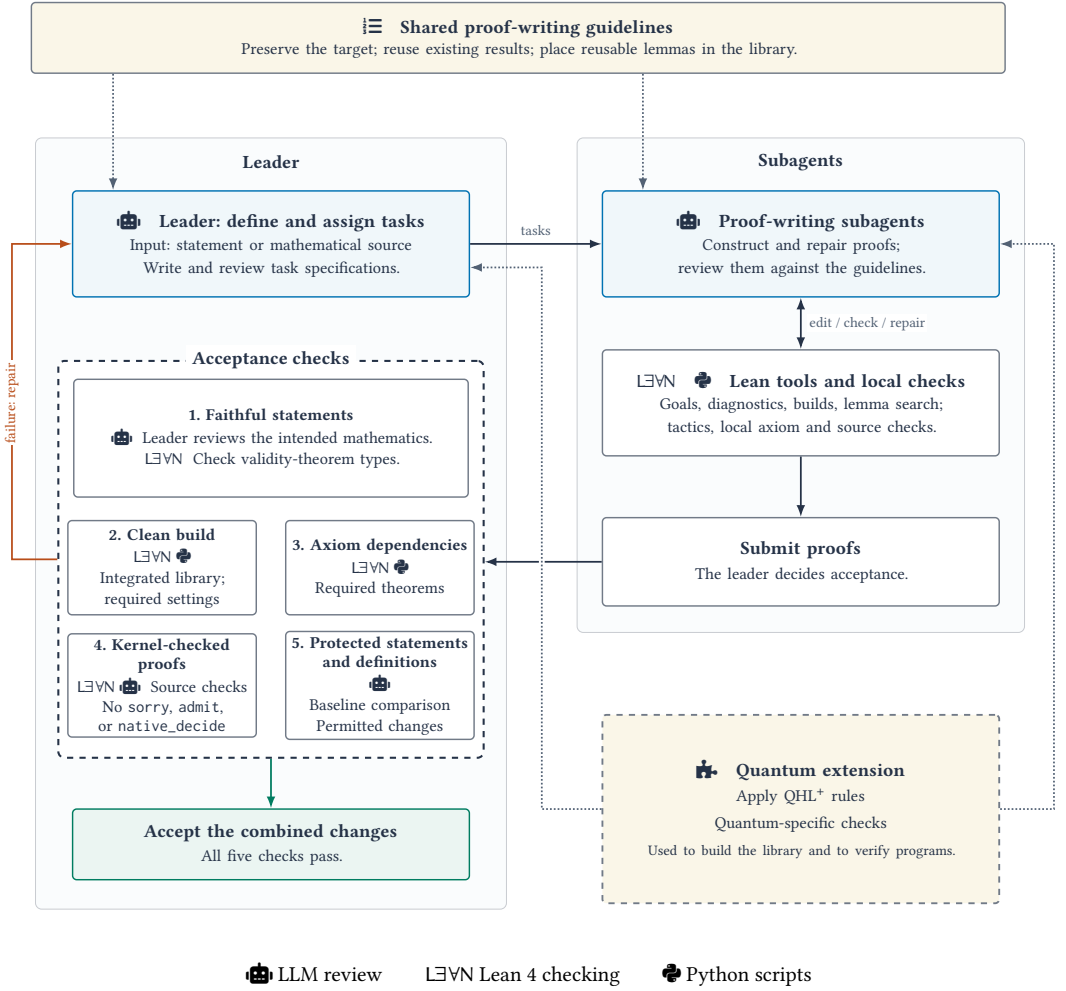

\section{Verified Program Examples}\label{sec:suite}

We present eleven verified program examples, constructed using the prover agent
(Section~\ref{sec:qpv}).

\subsection{Overview}\label{ssec:ex-overview}\label{ssec:ex-basic}

Table~\ref{tab:examples} summarizes the language features, verified properties, and subprogram reuse of these examples.
Examples marked S reuse the verified derivations of the subprograms in the
Reuse column. Each such subprogram is placed in the larger program by
renaming its quantum variables with an injection and relating its
classical state to that of the larger program.
All Hoare triples are verified for total correctness, which accounts for termination.

\begingroup
\footnotesize
\setlength{\tabcolsep}{3pt}
\setlength{\LTcapwidth}{\textwidth}
\begin{xltabular}{\linewidth}{@{}L{9.5em} L{4.3em} Y L{6.6em}@{}}
\caption{Overview of program examples. They illustrate the following program features: L (loop), NL (nested
loops), A (quantum arrays), P (classical inputs), C (conditionals),
M (measurement-controlled flow), and S (subprogram reuse).}\label{tab:examples} \\
\toprule
Example & Features & Verified property & Reuse \\
\midrule
\endfirsthead
\multicolumn{4}{c}{\small\itshape Table~\thetable{} (continued)} \\
\toprule
Example & Features & Verified property & Reuse \\
\midrule
\endhead
\bottomrule
\endfoot
\bottomrule
\endlastfoot
Bell state & -- & prepares $\ket{\Phi^+}$ from $\ket{00}$ & -- \\
Parallel Hadamard & L, A & prepares $\ket{+}^{\otimes n}$ for $n\ge1$ & -- \\
QFT & NL, A, P & implements the Fourier transform for $n\ge1$ & -- \\
Inverse QFT & NL, A, P & carries $\ket{\chi(n,x)}$ to $\ket{x}$ & -- \\
QPE & NL, A, P, S & recovers the $n$-bit eigenphase $\bar x/2^n$ of $U$ exactly & Parallel Hadamard, inverse QFT \\
Cluster state & L, A & prepares the 1D cluster state $\ket{C_n}$ for $n\ge1$ & -- \\
RUS & L, M, P & applies $V_{\RUS}$ with probability at least $1-4^{-N}$ & -- \\
Trotter & NL, A, P & for $N\ge1$ and $\varepsilon\ge\theta^2\|[H_X,H_{ZZ}]\|/(2N)$, every input $\ket\psi$ has output fidelity at least $1-\varepsilon$ with $e^{-i\theta H}\ket\psi$ & -- \\
Bravyi--Kitaev & NL, A, C & implements $\ket f\mapsto\ket{\beta_{2^x}f}$ & -- \\
Fermi--Hubbard LCU & NL, A, P & for $\lambda>0$, block-encodes $H/\lambda$ with success probability $\|H\ket\psi\|^2/\lambda^2$ & -- \\
Fermi--Hubbard QSP & NL, A, P, C, S & for $\mathrm{spec}(H)\subset(-\lambda,\lambda)$, block-encodes $P_\Phi(H/\lambda)\,\beta_{2^x}$ & Bravyi--Kitaev, Fermi\nobreakdash--Hubbard LCU \\
\end{xltabular}
\endgroup

In the following, we present the verification for the QFT, RUS, and Trotter programs introduced
in Section~\ref{sec:motivating}. Appendix~\ref{app:outlines} presents the
remaining eight program examples, with the three building blocks for simulating the
1D spinless Fermi--Hubbard model grouped in Appendix~\ref{ssec:ex-fermion}.

\subsection{Quantum Fourier Transform}\label{ssec:ex-qft}

We first derive the Hoare triple \eqref{eq:qfttriple} for the QFT program in Example~\ref{ssec:qft}.
The \hyperref[fig:qftoutline]{QFT proof outline} presents this derivation.
In a proof outline~\cite{OwickiGries1976,AptOlderog2009}, the program is annotated with assertions. Each program
statement, together with the assertions before and after it, forms a
Hoare triple. The right column gives the proof rules used in deriving each Hoare triple. (Rule-Seq) is used implicitly to sequentially combine Hoare triples.
The target Hoare triple is formed by the first and last
assertions in the proof outline.

We explain the notation as follows.
For simplicity of presentation, we write a pure state $\ket{\varphi}$ for its projector $\ketbra\varphi$ in Hoare triples and write quantum predicates directly as operator expressions. The invariant names are local to each example unless explicitly reused.
Let
$0.x[i{:}n{-}1] = \sum_{t=i}^{n-1} x[t]\,2^{\,i-t-1}$,
$\theta(i) = 2\pi \cdot 0.x[i{:}n{-}1]$, and $
  \ket{\psi(\alpha)} \;=\; \tfrac{1}{\sqrt{2}}\bigl(\ket{0} +
  e^{i\alpha}\ket{1}\bigr)$.
Then, the final Fourier basis state can be written as
\begin{equation}\label{eq:fourierket}
  \ket{\chi(n,x)}
  \;=\; \bigotimes_{j=0}^{n-1} \ket{\psi(\theta(n{-}1{-}j))}.
\end{equation}

The key components of the \hyperref[fig:qftoutline]{QFT proof outline} are
the three loops at lines 2, 5, and 12, where we use (Rule-Loop-tot).
The first loop invariant is $\{i\le n,\inv(i)\}$, where
\[
  \inv(i)=\bigotimes_{m=0}^{i-1}\ket{\psi(\theta(m))}
       \otimes\bigotimes_{m=i}^{n-1}\ket{x[m]},
\]
meaning the first $i$ quantum variables already carry their Fourier
phases. The second loop invariant is $\{i<n\wedge j\le n-1-i,\inv'(i,j)\}$, where
\[
  \inv'(i,j)=\bigotimes_{m=0}^{i-1}\ket{\psi(\theta(m))}\otimes
  \ket{\psi\!\left(x[i]\pi+\sum_{t=0}^{j-1}x[i{+}1{+}t]
    \frac{2\pi}{2^{t+2}}\right)}\otimes
  \bigotimes_{m=i+1}^{n-1}\ket{x[m]}.
\]
The third loop invariant is $\{i\le\lfloor n/2\rfloor,\inv''(i)\}$, where
\[
  \inv''(i)=\Bigl(\textstyle\prod_{m=0}^{i-1}
    \SWAP[q[m],q[n{-}1{-}m]]\Bigr)
    \bigotimes_{m=0}^{n-1}\ket{\psi(\theta(m))}.
\]

\begin{proofoutline}{fig:qftoutline}{Quantum Fourier transform}

\begin{outlinerows}

\qassert{0}
  {\mathit{true},\ \ket{x}}
  {}

\qline{1}{0}
  {i := 0}
  {(Axiom-Ass), (Rule-Conseq)}

\qassert{0}
  {i=0,\ \inv(i)}
  {$\inv(i)$: $i$ iterations completed; $\inv(0)=\ket{x}$}

\qline{2}{0}
  {\kw{while}\ i<n\ \kw{do}}
  {(Rule-Loop-tot), $t=n-i$}

\qassert{1}
  {i<n,\ \inv(i)}
  {}

\qline{3}{1}
  {H[q[i]]}
  {(Axiom-Uni)}

\qassert{1}
  {i<n,\ \inv'(i,0)}
  {}

\qline{4}{1}
  {j := 0}
  {(Axiom-Ass), (Rule-Conseq)}

\qassert{1}
  {i<n \wedge j=0,\ \inv'(i,j)}
  {$\inv'(i,j)$: after $j$ controlled phase gates}

\qline{5}{1}
  {\kw{while}\ j<n-1-i\ \kw{do}}
  {(Rule-Loop-tot), $t=(n-1-i)-j$}

\qassert{2}
  {i<n \wedge j<n-1-i,\ \inv'(i,j)}
  {}

\qline{6}{2}
  {\cR(2\pi/2^{\,j+2})[q[i],q[i+1+j]]}
  {(Axiom-Uni)}

\qassert{2}
  {i<n \wedge j<n-1-i,\ \inv'(i,j+1)}
  {}

\qline{7}{2}
  {j := j+1}
  {(Axiom-Ass), (Rule-Conseq)}

\qassert{2}
  {i<n\wedge j\le n-1-i,\ \inv'(i,j)}
  {}

\qline{8}{1}
  {\kw{od}}
  {}

\qassert{1}
  {i<n \wedge j=n-1-i,\ \inv'(i,j)}
  {}

\qassert{1}
  {i<n,\ \inv(i+1)}
  {(Rule-Conseq)}

\qline{9}{1}
  {i := i+1}
  {(Axiom-Ass), (Rule-Conseq)}

\qassert{1}
  {i\le n,\ \inv(i)}
  {}

\qline{10}{0}
  {\kw{od}}
  {}

\qassert{0}
  {i=n,\ \inv(n)}
  {}

\qassert{0}
  {\mathit{true},\
   \displaystyle\bigotimes_{m=0}^{n-1}\ket{\psi(\theta(m))}}
  {(Rule-Conseq)}

\qline{11}{0}
  {i := 0}
  {(Axiom-Ass), (Rule-Conseq)}

\qassert{0}
  {i=0,\ \inv''(i)}
  {$\inv''(i)$: the first $i$ swaps done}

\qline{12}{0}
  {\kw{while}\ i<\lfloor n/2\rfloor\ \kw{do}}
  {(Rule-Loop-tot), $t=\lfloor n/2\rfloor-i$}

\qassert{1}
  {i<\lfloor n/2\rfloor,\ \inv''(i)}
  {}

\qline{13}{1}
  {\SWAP[q[i],q[n-1-i]]}
  {(Axiom-Uni)}

\qassert{1}
  {i<\lfloor n/2\rfloor,\ \inv''(i+1)}
  {}

\qline{14}{1}
  {i := i+1}
  {(Axiom-Ass), (Rule-Conseq)}

\qassert{1}
  {i\le\lfloor n/2\rfloor,\ \inv''(i)}
  {}

\qline{15}{0}
  {\kw{od}}
  {}

\qassert{0}
  {i=\lfloor n/2\rfloor,\
   \inv''(\lfloor n/2\rfloor)}
  {}

\qassert{0}
  {\mathit{true},\ \ket{\chi(n,x)}}
  {(Rule-Conseq)}

\end{outlinerows}
\end{proofoutline}

At line 1, $\inv(0)=\ket{x}$. The Hadamard gate at
line 3 converts $\inv(i)$ into $\inv'(i,0)$,
which is turned into $\inv'(i,j)$ by the second $\mathbf{while}$ loop. Inside the loop, at line 6, the controlled rotation adds the phase contributed by $x[i+1+j]$ through $\cR(2\pi/2^{j+2})[q[i],q[i+1+j]]\inv'(i,j)=\inv'(i,j+1)$.
From $\inv'(i,n-1-i)=\inv(i+1)$, the first $\mathbf{while}$ loop ends at line 10 with $\inv(n)$.
The third $\mathbf{while}$ loop reverses the qubit order by $\SWAP[q[i],q[n-1-i]]\inv''(i)=\inv''(i+1)$.

The Hoare triple~\eqref{eq:qfttriple} determines the output of each
computational basis input up to a phase. To establish correctness on arbitrary inputs,
we also derive the following Hoare triple for $n\geq1$:
\begin{equation}
  \models_{\mathit{tot}}\quad
  \bigl\{\ket{+}^{\otimes n}\bigr\}\ \QFT_n\
  \bigl\{\ket{0}^{\otimes n}\bigr\}.
  \label{eq:qft-uniform}
\end{equation}
Appendix~\ref{ssec:qft-phase} derives~\eqref{eq:qft-uniform} and shows that
\eqref{eq:qfttriple} and~\eqref{eq:qft-uniform} together determine the output
on every input up to a common global phase.

\subsection{Repeat-Until-Success}\label{ssec:ex-rus}

We next derive the Hoare triples~\eqref{eq:rustriple}
and~\eqref{eq:russuccess} for the RUS program in Example~\ref{ssec:rus}.
The \hyperref[fig:rusoutline]{RUS proof outline} presents the derivation
of~\eqref{eq:rustriple}.

We explain the notation as follows.
The gate $W_{\RUS}$ is explicitly given by~\cite[Fig.~8]{PaetznickSvore2014}:
\begin{equation}\label{eq:wrus}
W_{\RUS} \;=\; (H \otimes I)\,(T \otimes I)\,\CNOT[a,d]\,
  (H \otimes I)\,\CNOT[a,d]\,(T \otimes I)\,(H \otimes I),
\end{equation}
where $T$ is the phase-$\frac{\pi}{4}$ gate, and $a$ is the first tensor factor and the control of each CNOT. It satisfies, for all $\ket{\psi_0}$,
\begin{equation}\label{eq:branch}
 W_{\RUS}\bigl(\ket0\otimes\ket{\psi_0}\bigr)
 =\alpha\,\ket0\otimes V_{\RUS}\ket{\psi_0}+\beta\,\ket1\otimes\ket{\psi_0},
\end{equation}
where $|\alpha|^2=3/4$ and $|\beta|^2=1/4$.
Let $\ket{v_b}=\ket b\otimes\ket{\psi_0}$ for $b\in\{0,1\}$. Then,
$I\otimes\ketbra{\psi_0}=\ketbra{v_0}+\ketbra{v_1}$, and~\eqref{eq:branch} means
$W_{\RUS}\ket{v_0}=\alpha\ket w+\beta\ket{v_1}$.

The key component of the \hyperref[fig:rusoutline]{RUS proof outline} is
the $\mathbf{while}$ loop at line 3, where we use (Rule-Loop-tot).
The loop invariant is $\{k\le N,\inv(m)\}$, with $\inv(m)$ defined
in~\eqref{eq:rus-invariant}, meaning that the data qubit carries
$V_{\RUS}\ket{\psi_0}$ if $m=0$ and is unchanged if $m=1$,
so that each attempt restarts from $\ket{v_0}$.

\begin{proofoutline}{fig:rusoutline}{Repeat-until-success}
\begin{outlinerows}

\qassert{0}
  {N\ge1,\ I \otimes \ketbra{\psi_0}}
  {}

\qline{1}{0}
  {m := 1}
  {(Axiom-Ass), (Rule-Conseq)}

\qassert{0}
  {N\ge1\wedge m = 1,\ \ketbra{v_0} + \ketbra{v_1}}
  {}

\qline{2}{0}
  {k := 0}
  {(Axiom-Ass), (Rule-Conseq)}

\qassert{0}
  {N\ge1\wedge m = 1 \wedge k = 0,\ \inv(m)}
  {}

\qline{3}{0}
  {\kw{while}\ (m = 1 \wedge k < N)\ \kw{do}}
  {(Rule-Loop-tot), $t = N{-}k$}

\qassert{1}
  {m = 1 \wedge k < N,\ \ketbra{v_0} + \ketbra{v_1}}
  {}

\qline{4}{1}
  {a := \ket{0}}
  {(Axiom-Init)}

\qassert{1}
  {m = 1 \wedge k < N,\ \ket{v_0}}
  {}

\qline{5}{1}
  {W_{\RUS}[a,d]}
  {(Axiom-Uni)}

\qassert{1}
  {m = 1 \wedge k < N,\ \alpha\ket{w} + \beta\ket{v_1}}
  {}

\qline{6}{1}
  {m := M[a]}
  {(Axiom-Meas), (Rule-Conseq)}

\qassert{1}
  {k < N,\ (1{-}m)\,\ketbra{w} + m\,\ketbra{v_1}}
  {}

\qline{7}{1}
  {k := k + 1}
  {(Axiom-Ass), (Rule-Conseq)}

\qassert{1}
  {k\le N,\ \inv(m)}
  {}

\qline{8}{0}
  {\kw{od}}
  {}

\qassert{0}
  {k\le N\wedge\neg(m=1\wedge k<N),\ \inv(m)}
  {}
\qassert{0}
  {m = 0 \vee k \geq N,\ \inv(m)}
  {(Rule-Conseq)}

\end{outlinerows}
\end{proofoutline}

At line 6, (Axiom-Meas) sums over both outcomes of measuring $a$, which updates
the classical variable $m$ controlling the loop. It converts the postcondition
$(1-m)\ketbra w+m\ketbra{v_1}$ of line 6 into
$\ketbra w+\ketbra{v_1}\ge\ketbra{\alpha w+\beta v_1}$.
The loop ends at line 8 with $m=0\vee k\ge N$, which is the postcondition
of~\eqref{eq:rustriple}.

We then derive the success bound~\eqref{eq:russuccess} in the
\hyperref[fig:russuccessoutline]{RUS success proof outline}, using the loop
invariant $\{k\le N\wedge(m=0\vee k<N),\ \inv'(m,k)\}$, where
\[
 \inv'(m,k)=(1-m)\ketbra w+
 m\,p_k\bigl(\ketbra{v_0}+\ketbra{v_1}\bigr),\qquad
 p_k=1-|\beta|^{2(N-k)}.
\]
For $m=1$, $p_k$ is the probability of at least one success within the
remaining $N-k$ attempts. At loop exit, the classical part excludes the
exhausted branch $m=1\wedge k=N$, so only successful outputs are counted
in~\eqref{eq:russuccess}.

\begin{proofoutline}{fig:russuccessoutline}{Repeat-until-success: success probability}
\begin{outlinerows}
\qassert{0}
 {N\ge1,\ (1-4^{-N})(I\otimes\ketbra{\psi_0})}{}
\qline{1}{0}{m:=1}{(Axiom-Ass), (Rule-Conseq)}
\qassert{0}{N\ge1\wedge m=1,\ \inv'(m,0)}{}
\qline{2}{0}{k:=0}{(Axiom-Ass), (Rule-Conseq)}
\qassert{0}{N\ge1\wedge m=1\wedge k=0,\ \inv'(m,k)}{}
\qline{3}{0}{\kw{while}\ (m=1\wedge k<N)\ \kw{do}}
 {(Rule-Loop-tot), $t=N-k$}
\qassert{1}{m=1\wedge k<N,\
 p_k(\ketbra{v_0}+\ketbra{v_1})}{}
\qline{4}{1}{a:=\ket0}{(Axiom-Init), (Rule-Conseq)}
\qassert{1}{m=1\wedge k<N,\
 W_{\RUS}^{\dagger}\bigl(\ketbra w+p_{k+1}\ketbra{v_1}\bigr)W_{\RUS}}{}
\qline{5}{1}{W_{\RUS}[a,d]}{(Axiom-Uni)}
\qassert{1}{m=1\wedge k<N,\ \ketbra w+p_{k+1}\ketbra{v_1}}{}
\qline{6}{1}{m:=M[a]}
 {(Axiom-Meas), (Rule-Mask), (Rule-Conseq)}
\qassert{1}{k<N\wedge(m=0\vee k+1<N),\ \inv'(m,k+1)}{}
\qline{7}{1}{k:=k+1}{(Axiom-Ass), (Rule-Conseq)}
\qassert{1}{k\le N\wedge(m=0\vee k<N),\ \inv'(m,k)}{}
\qline{8}{0}{\kw{od}}{}
\qassert{0}{k\le N\wedge(m=0\vee k<N)\wedge\neg(m=1\wedge k<N),\ \inv'(m,k)}{}
\qassert{0}{m=0,\ \ketbra w}{(Rule-Conseq)}
\end{outlinerows}
\end{proofoutline}

The quantum invariant is preserved since
$|\alpha|^2+|\beta|^2p_{k+1}=p_k$ for $k<N$, that is, the
probability $|\alpha|^2$ of $m=0$ in the current attempt plus the probability
$|\beta|^2p_{k+1}$ of $m=1$ followed by a later success.
At line 6, when $k+1=N$, the outcome $m=1$ would violate $m=0\vee k+1<N$.
We therefore apply (Axiom-Meas) with the mask $m=0\vee k+1<N$ on
$\inv'(m,k+1)$, and the precondition is still the assertion after line 5
because $p_N=0$.
Then (Rule-Mask) converts the mask into the classical postcondition of line 6.

In both outlines, $k\le N$ keeps the bound function $N-k$ nonnegative, and the loop
body decreases it under $k<N$.

\subsection{The Trotter Product-Formula Circuit}\label{ssec:ex-ising}\label{ssec:ex-trotter}

We then derive the Hoare triple~\eqref{eq:trottertriple} for the Trotter
program in Example~\ref{ssec:trotter}. The
\hyperref[fig:trotteroutline]{Trotter proof outline} presents this
derivation.

We explain the notation as follows.
Let $A_j(\delta)$ and $B_j(\delta)$ be the products of the first $j$ rotations
applied by the second and third $\mathbf{while}$ loops, respectively:
\begin{equation}\label{eq:trotter-partial-products}
\begin{aligned}
 A_0(\delta)&=B_0(\delta)=I,\\
 A_{j+1}(\delta)&=R_x(2a[j]\delta)[q[j]]\,A_j(\delta),\\
 B_{j+1}(\delta)&=R_{zz}(2b[j]\delta)[q[j],q[j+1]]\,B_j(\delta).
\end{aligned}
\end{equation}
Let $A(\delta)=A_n(\delta)$ and $B(\delta)=B_{n-1}(\delta)$, which means
$A(\delta)=e^{-i\delta H_X}$ and $B(\delta)=e^{-i\delta H_{ZZ}}$. Each
iteration of the first $\mathbf{while}$ loop therefore applies
$W=B(\delta)A(\delta)=e^{-i\delta H_{ZZ}}e^{-i\delta H_X}$.
Let $\Phi_\varepsilon$ denote the classical precondition
$N\ge1\wedge\varepsilon\ge\theta^2\|[H_X,H_{ZZ}]\|/(2N)$
of~\eqref{eq:trottertriple}. Define $c=(1+\varepsilon)^{-1}$.
Since $a,b,\theta,N,\varepsilon\notin\mathit{change}(\Trotter_n)$, $\Phi_\varepsilon$ remains
in every invariant and supplies the assumptions of the final (Rule-Conseq) step.

The key components of the \hyperref[fig:trotteroutline]{Trotter proof
outline} are the three loops at lines 2, 4, and 9, where we use
(Rule-Loop-tot). The first loop invariant is
$\{\Phi_\varepsilon\wedge i\le N,\ \inv(i)\}$, where
\[
 \inv(i)=c\ketbra{W^i\psi},
\]
meaning that the first $i$ iterations of the first $\mathbf{while}$ loop are completed. The second loop invariant is
$\{\Phi_\varepsilon\wedge i<N\wedge j\le n,\ \inv'(i,j)\}$, where
\[
 \inv'(i,j)=c\ketbra{A_j(\delta)W^i\psi},
\]
meaning that the first $j$ rotations $R_x$ of the current iteration are
applied. The third loop invariant is
$\{\Phi_\varepsilon\wedge i<N\wedge j\le n-1,\ \inv''(i,j)\}$, where
\[
 \inv''(i,j)=c\ketbra{B_j(\delta)A(\delta)W^i\psi},
\]
meaning that all rotations $R_x$ and the first $j$ rotations $R_{zz}$ of
the current iteration are applied.

\begin{proofoutline}{fig:trotteroutline}{Trotter product-formula circuit}
\begin{outlinerows}

\qassert{0}
  {\Phi_\varepsilon,\ c\ketbra\psi}
  {}

\qline{1}{0}
  {i := 0}
  {(Axiom-Ass), (Rule-Conseq)}

\qassert{0}
  {\Phi_\varepsilon\wedge i = 0,\ \inv(i)}
  {}

\qline{2}{0}
  {\kw{while}\ i < N\ \kw{do}}
  {(Rule-Loop-tot), $t = N{-}i$}

\qassert{1}
  {\Phi_\varepsilon\wedge i < N,\ \inv(i)}
  {}

\qline{3}{1}
  {j := 0}
  {(Axiom-Ass), (Rule-Conseq)}

\qassert{1}
  {\Phi_\varepsilon\wedge i < N \wedge j = 0,\ \inv'(i,j)}
  {}

\qline{4}{1}
  {\kw{while}\ j < n\ \kw{do}}
  {(Rule-Loop-tot), $t = n{-}j$}

\qassert{2}
  {\Phi_\varepsilon\wedge i < N \wedge j < n,\ \inv'(i,j)}
  {}

\qline{5}{2}
  {R_x(2a[j]\delta)[q[j]]}
  {(Axiom-Uni)}

\qassert{2}
  {\Phi_\varepsilon\wedge i < N \wedge j < n,\ \inv'(i,j{+}1)}
  {}

\qline{6}{2}
  {j := j + 1}
  {(Axiom-Ass), (Rule-Conseq)}

\qassert{2}
  {\Phi_\varepsilon\wedge i < N\wedge j\le n,\ \inv'(i,j)}
  {}

\qline{7}{1}
  {\kw{od}}
  {}

\qassert{1}
  {\Phi_\varepsilon\wedge i < N \wedge j = n,\ \inv'(i,n)}
  {}

\qassert{1}
  {\Phi_\varepsilon\wedge i < N,\ c\ketbra{A(\delta)W^i\psi}}
  {(Rule-Conseq)}

\qline{8}{1}
  {j := 0}
  {(Axiom-Ass), (Rule-Conseq)}

\qassert{1}
  {\Phi_\varepsilon\wedge i < N \wedge j = 0,\ \inv''(i,j)}
  {}

\qline{9}{1}
  {\kw{while}\ j < n{-}1\ \kw{do}}
  {(Rule-Loop-tot), $t = (n{-}1){-}j$}

\qassert{2}
  {\Phi_\varepsilon\wedge i < N \wedge j < n{-}1,\ \inv''(i,j)}
  {}

\qline{10}{2}
  {R_{zz}(2b[j]\delta)[q[j], q[j{+}1]]}
  {(Axiom-Uni)}

\qassert{2}
  {\Phi_\varepsilon\wedge i < N \wedge j < n{-}1,\ \inv''(i,j{+}1)}
  {}

\qline{11}{2}
  {j := j + 1}
  {(Axiom-Ass), (Rule-Conseq)}

\qassert{2}
  {\Phi_\varepsilon\wedge i < N\wedge j\le n-1,\ \inv''(i,j)}
  {}

\qline{12}{1}
  {\kw{od}}
  {}

\qassert{1}
  {\Phi_\varepsilon\wedge i < N \wedge j = n{-}1,\ \inv''(i,n{-}1)}
  {}

\qassert{1}
  {\Phi_\varepsilon\wedge i < N,\ \inv(i{+}1)}
  {(Rule-Conseq)}

\qline{13}{1}
  {i := i + 1}
  {(Axiom-Ass), (Rule-Conseq)}

\qassert{1}
  {\Phi_\varepsilon\wedge i\le N,\ \inv(i)}
  {}

\qline{14}{0}
  {\kw{od}}
  {}

\qassert{0}
  {\Phi_\varepsilon\wedge i = N,\ \inv(i)}
  {}

\qassert{0}
  {\Phi_\varepsilon,\ c\ketbra{W^N\psi}}
  {(Rule-Conseq)}

\qassert{0}
  {\mathit{true},\ c\bigl(\ketbra{U_\theta\psi}+\varepsilon I\bigr)}
  {(Rule-Conseq)}

\end{outlinerows}
\end{proofoutline}
At line 1, $\inv(0)=c\ketbra\psi$. By~\eqref{eq:trotter-partial-products},
the rotation at line 5 converts $\inv'(i,j)$ into $\inv'(i,j+1)$, and the
rotation at line 10 converts $\inv''(i,j)$ into $\inv''(i,j+1)$.
From $B(\delta)A(\delta)=W$, we have $\inv''(i,n-1)=\inv(i+1)$, and the first
$\mathbf{while}$ loop ends at line 14 with $\inv(N)=c\ketbra{W^N\psi}$.

Only the last (Rule-Conseq) step uses the condition on $\varepsilon$ in $\Phi_\varepsilon$.
It relies on the first-order Trotter error bound~\cite{Childs2021TrotterError}
\begin{equation}\label{eq:trotter-error}
 \|W^N-U_\theta\|\le\frac{\theta^2}{2N}\|[H_X,H_{ZZ}]\|,
\end{equation}
which sums the one-step error $\delta^2\|[H_X,H_{ZZ}]\|/2$ over the $N$
iterations. Under $\Phi_\varepsilon$, \eqref{eq:trotter-error} gives
$\|W^N\psi-U_\theta\psi\|\le\varepsilon$. Hence
$\ketbra{W^N\psi}\le\ketbra{U_\theta\psi}+\varepsilon I$, since projectors
onto unit vectors satisfy
$\|\ketbra u-\ketbra v\|\le\|u-v\|$.

\section{Related Work}\label{sec:related}

\paragraph{Quantum Hoare logics.}
Quantum Hoare logic (QHL)~\cite{Ying2012QHL} provides sound and relatively
complete proof systems for partial and total correctness of quantum
while-programs without classical variables. The applied quantum Hoare
logic~\cite{Zhou2019aQHL} restricts preconditions and postconditions to
projections, which simplifies the proof rules for larger programs. For
programs with classical variables, the logic of~\cite{FengYing2021} uses
families of Hermitian operators indexed by classical states as
preconditions and postconditions. \qhl{}~\cite{Ying2025} instead pairs a
classical first-order formula with a single quantum predicate formula
parameterized by classical variables. Specifications thus remain compact,
and classical first-order reasoning applies directly to their classical
parts.
\qhl{} also supports quantum arrays and gates whose indices and
parameters are given by classical variables. Earlier quantum Hoare logics
do not support these features~\cite{Ying2025}.

\paragraph{Mechanized quantum program logics.}
The QHL prover~\cite{Liu2019QHLProver} formalizes a finite-dimensional
partial-correctness fragment of QHL in Isabelle/HOL and proves it sound
and complete. CoqQ~\cite{Zhou2023} verifies partial and total correctness
of quantum while-programs in Coq, with a machine-checked soundness proof.
Neither supports classical program variables.
CoqQLR~\cite{Wu2024CoqQLR,wu2025localreasoningprobabilisticbehaviour} formalizes in Coq a logic for local reasoning
about probabilistic properties of classical--quantum programs, with
soundness for trace-preserving programs. The
qrhl-tool~\cite{Unruh2019qrhl} supports quantum relational Hoare logic
(qRHL) in Isabelle/HOL, with soundness proved on paper.

\paragraph{Verification of circuits and hybrid programs.}
QWIRE~\cite{Paykin2017QWIRE} embeds a typed circuit language in Coq with
dynamic lifting of measurement results. SQIR/VOQC~\cite{Hietala2021SQIR}
verifies circuit optimizations in Coq, and a certified implementation of
Shor's algorithm~\cite{Peng2023Shor} builds on SQIR.
Qbricks~\cite{Chareton2021Qbricks} verifies circuit-building programs in
Why3. Qafny~\cite{Li2024Qafny} has a type-guided quantum proof system and
automates verification by translation into Dafny.
Table~\ref{tab:related} summarizes these systems.

\begin{table}[!htbp]
  \centering
  \small
  \setlength{\tabcolsep}{4pt}
  \caption{Quantum program verification tools: QEDIA in context.}
  \label{tab:related}
  \begin{tabularx}{\linewidth}{@{}L{13.6em} L{5.6em} Y L{11.6em}@{}}
    \toprule
    Tool & Assistant & Program model & Classical data/control \\
    \midrule
    CoqQ~\cite{Zhou2023} & Coq & quantum while & no classical variables \\
    QHL prover~\cite{Liu2019QHLProver} & Isabelle/HOL & quantum while & no classical variables \\
    CoqQLR~\cite{wu2025localreasoningprobabilisticbehaviour} & Coq & classical--quantum while & classical variables \\
    qrhl-tool~\cite{Unruh2019qrhl} & Isabelle/HOL & relational (program pairs) & classical variables \\
    QWIRE~\cite{Paykin2017QWIRE} & Coq & circuits & classical wires, dynamic lifting \\
    SQIR/VOQC~\cite{Hietala2021SQIR} & Coq & circuits & measurement branches \\
    Qbricks~\cite{Chareton2021Qbricks} & Why3 & circuit-building programs & circuit construction \\
    Qafny~\cite{Li2024Qafny} & Coq/Dafny & quantum programs (via Dafny) & immutable bindings \\
    \midrule
    QEDIA \textbf{(this work)} & Lean~4 & classical--quantum while & classical variables indexing qubits and gate parameters \\
    \bottomrule
  \end{tabularx}
\end{table}

To our knowledge, QEDIA is the first mechanization of a quantum Hoare
logic that supports quantum arrays and gates parameterized by classical
program variables, with soundness
proved for both partial and total correctness.

\section{Conclusion and Future Work}\label{sec:conclusion}
Quantum programs interleave classical and quantum computation. Verifying
them in a proof assistant requires handling this interplay consistently,
from the mathematics of quantum mechanics to the proof rules. QEDIA does
so for \qhl{} in Lean~4 and proves the logic sound for both partial and
total correctness. On this
foundation, we verified programs ranging from the quantum Fourier
transform and repeat-until-success circuits to Trotter-based Hamiltonian
simulation and quantum signal processing for the Fermi--Hubbard model. LLM agents built the library and these proofs under human
supervision, and Lean's kernel checks every proof.

In the following, we list a few topics for future work.
\begin{itemize}[leftmargin=1.4em]
\item \textbf{Quantum error correction.} Quantum error correction
  underpins fault-tolerant quantum computation. Its syndrome measurement
  outcomes determine, through classical decoding, which corrections are
  applied. This makes it
  a natural next target for QEDIA.
\item \textbf{Quantum recursion.} Quantum recursive
  programs~\cite{YingZhang2026,ZhangYing2025QRM} combine quantum control
  flow with recursive procedure calls. Extending QEDIA to them requires
  handling this combination, since \textsf{qWhile}$^+$ has neither quantum
  control flow nor procedure calls.
\item \textbf{Parallel and distributed programs.} Hoare-style logics have
  been developed for parallel~\cite{Ying2018Parallel} and
  distributed~\cite{FengLiYing2022Distributed} quantum programs.
  Mechanizing them raises further challenges, such as entanglement between
  component programs and the intertwining of nondeterminism from
  measurements and from parallelism~\cite{Ying2018Parallel}.
\end{itemize}

\begin{acks}
\paragraph{Use of AI}
The Lean formalization and proofs were constructed by LLM agents under
human supervision (Section~\ref{sec:qpv}), using OpenAI's GPT-5.5, GPT Sol (versions 5.6, 6,
and 6.1), and GPT-6 Astra, as well as Anthropic's Claude Opus (versions
4.6, 4.8, and 5.5) and Claude Fable (versions 5 and 5.1). We also used
AI assistance during the later stages of the project to organize data
and improve the clarity and language of the manuscript. The formal
proofs are checked by the Lean~4 kernel. The authors take responsibility
for the specifications, interpretation of the results, and final content
of the paper.
\end{acks}

\bibliographystyle{ACM-Reference-Format}
\bibliography{refs}

\appendix

\section{Details of the Formalization}
\label{app}

\subsection{Deriving the Original Measurement Axiom}\label{app:meas}

We derive the measurement axiom of~\cite{Ying2025}, which treats one
outcome at a time, from (Axiom-Meas) and (Rule-Mask). The axiom reads
\[
 \{\varphi[y/x],\ F_M(y)[\bar q](A[y/x])\}\ x:=M[\bar q]\
 \{\varphi\wedge x=y,\ A\}.
\]
Here, $y\notin\mathit{free}(\varphi)\cup\mathit{cv}(A)\cup\{x\}$ is a
classical variable and $F_M(y)=\{M_y^\dagger\}$. The following proof outline
derives it, provided $\varphi[y/x]\models\forall x.\mathrm{WD}(A)$.

\begin{proofoutline}{fig:measoutline}{Single-outcome measurement}
\begin{outlinerows}

\qassert{0}
  {\varphi[y/x],\ F_M(y)[\bar q](A[y/x])}
  {}

\qline{1}{0}
  {x := M[\bar q]}
  {(Axiom-Meas), (Rule-Conseq)}

\qassert{0}
  {\varphi\vee x\neq y,\ (x=y)\cdot A}
  {}

\qassert{0}
  {(\varphi\vee x\neq y)\wedge x=y,\ A}
  {(Rule-Mask)}

\qassert{0}
  {\varphi\wedge x=y,\ A}
  {(Rule-Conseq)}

\end{outlinerows}
\end{proofoutline}

At line 1, (Axiom-Meas) gives the precondition
$\{\forall x.(\varphi\vee x\neq y),\ F_M[\bar q](\{(j=y)\cdot A[j/x]\})\}$.
Since $y$ is fresh, its classical part is equivalent to $\varphi[y/x]$.
In each classical state $\sigma$, the masks keep only the term for
$j=\sigma(y)$, so its quantum part has the same interpretation as
$F_M(y)[\bar q](A[y/x])$. Since $\varphi[y/x]\models\forall x.\mathrm{WD}(A)$, both quantum parts are
well-defined at the same classical states satisfying $\varphi[y/x]$.
(Rule-Conseq) then replaces one
precondition by the other.

\subsection{The RUS Program and Specification in Lean}\label{app:listings}

The following excerpts illustrate how the RUS program of
Figure~\ref{fig:programs}(a) and its success specification~\eqref{eq:russuccess}
are represented in Lean. Programs have type \texttt{Prog}
(Section~\ref{ssec:lean-programs}). The classical-state type \texttt{Store} ($\Sigma$ in Table~\ref{tab:lean-map})
represents $(N,m,k)$, with Boolean values \texttt{false} and \texttt{true}
encoding the outcomes $0$ and $1$. Surrounding namespace declarations
and theorem proofs are omitted.
\begin{lstlisting}[style=leanexcerpt]
noncomputable def body
    (W : UnitaryOp (EuclideanSpace ℂ (Bool × Bool))) : Prog Store D :=
  .seq (.init fun _ => regA.sig) (.seq (.unitary 2 rAT (gateW W))
    (.seq (.measure 2 1 rA measA updM) (.assign incK)))

noncomputable def prog
    (W : UnitaryOp (EuclideanSpace ℂ (Bool × Bool))) : Prog Store D :=
  .seq (.assign initM) (.seq (.assign initK) (.while guard (body W)))
\end{lstlisting}
Here, \texttt{initM}, \texttt{initK}, and \texttt{incK} implement
$m:=1$, $k:=0$, and $k:=k+1$, and \texttt{guard} is $m=1\wedge k<N$.
The expression \texttt{regA.sig} identifies the ancilla to reset;
\texttt{rA} and \texttt{rAT} select $a$ and $(a,d)$.
The operation \texttt{gateW W} applies $W$ to that pair,
\texttt{measA} is the computational-basis measurement, and
\texttt{updM} writes its outcome to $m$. Thus the measurement outcome
controls whether the loop executes again.

Here, \texttt{wPS} and \texttt{vPS}
are $W_{\RUS}$ and $V_{\RUS}$ of Example~\ref{ssec:rus}.
\begin{lstlisting}[style=leanexcerpt]
noncomputable def boundedRusClassicalProg : Prog ClassicalInput.Store D :=
  ClassicalInput.prog wPS

theorem boundedRusClassicalSuccessSourceValid (hψ₀ : ‖ψ₀‖ = 1) :
    validTot (ClassicalInput.preQ betaPS norm_betaPS_le_one ψ₀ hψ₀).toAssn
      boundedRusClassicalProg (ClassicalInput.postQ vPS ψ₀ hψ₀).toAssn := …
\end{lstlisting}
To identify its precondition and postcondition, we use the following local names:
\begin{lstlisting}[style=leanexcerpt]
local notation "Pre" =>
  ClassicalInput.preQ betaPS norm_betaPS_le_one ψ₀ hψ₀
local notation "Post" => ClassicalInput.postQ vPS ψ₀ hψ₀
\end{lstlisting}
At a classical state $\sigma=(N,m,k)$, their components satisfy
\[
\begin{aligned}
 (\texttt{Pre}).\texttt{classical}\ \sigma &\;\longleftrightarrow\; N\ge1,\\
 (\texttt{Post}).\texttt{classical}\ \sigma &\;\longleftrightarrow\; m=0,\\
 (\texttt{Pre}).\texttt{defined}\ \sigma
   &\;\longleftrightarrow\;\mathit{true}
   \;\longleftrightarrow\;(\texttt{Post}).\texttt{defined}\ \sigma,\\
 (\texttt{Pre}).\texttt{quantum.op}\ \sigma
   &= (1-4^{-N})(I\otimes\ketbra{\psi_0}),\\
 (\texttt{Post}).\texttt{quantum.op}\ \sigma &= \ketbra w.
\end{aligned}
\]
Here, $\ket w$ is defined in~\eqref{eq:rus-invariant}. The coefficient in the precondition is
established by \texttt{classicalSuccessPreQ\_op}, using the proved
failure-amplitude identity
\begin{lstlisting}[style=leanexcerpt]
theorem normSq_betaPS : ‖betaPS‖ ^ 2 = 1 / 4 := …
\end{lstlisting}
By the third line of the display, \texttt{toAssn} masks their quantum
predicates only by the displayed classical conditions (Section~\ref{ssec:lean-assertions}).
The validity theorem is~\eqref{eq:russuccess}; it follows from the corresponding
\texttt{SourceDerivTot} derivation by Theorem~\ref{thm:sound}.
Its proof outline is
in Section~\ref{ssec:ex-rus}.

\section{Further Verified Examples}\label{app:outlines}
This appendix presents the remaining eight examples of
Table~\ref{tab:examples}, after extending the QFT example of
Section~\ref{ssec:ex-qft} to arbitrary inputs. The proof outlines follow
the conventions of Section~\ref{ssec:ex-qft}.

\subsection{The QFT on Arbitrary Inputs}\label{ssec:qft-phase}

We derive the Hoare triple~\eqref{eq:qft-uniform} by the
\hyperref[fig:qftoutline]{QFT proof outline}, with the quantum parts of its
three loop invariants replaced by
$\inv(i)=\ket0^{\otimes i}\otimes\ket+^{\otimes(n-i)}$,
$\inv'(i,j)=\inv(i+1)$, and $\inv''(i)=\ket0^{\otimes n}$. At line 3, $H[q[i]]\inv(i)=\inv(i+1)=\inv'(i,0)$. At line 6, the
controlled rotation fixes $\inv'(i,j)$ because $q[i]$ is in $\ket0$. At
line 13, the swap fixes $\inv''(i)$. The first loop starts from
$\inv(0)=\ket+^{\otimes n}$ and ends with $\inv(n)=\inv''(0)$.

Let $V$ be the unitary applied by $\QFT_n$. For an input state $\ket\psi$,
we use the same outline with the loop invariants $\inv(i)$, $\inv'(i,j)$,
and $\inv''(i)$ taken to be the states obtained from $\ket\psi$ by the gates
of the outline up to the corresponding point. Then the verification
conditions at lines 3, 6, and 13 hold by definition, and
$\inv''(\lfloor n/2\rfloor)=V\ket\psi$. This establishes
\begin{equation}\label{eq:qft-circuit-output}
  \models_{\mathit{tot}}\quad
  \bigl\{\ket{\psi}\bigr\}\ \QFT_n\
  \bigl\{V\ket{\psi}\bigr\}.
\end{equation}

For every program $P$ and states $\ket w$, $\ket u$, and $\ket v$,
\[
 \models_{\mathit{tot}}\{\ket w\}\ P\ \{\ket u\}
 \ \text{ and }\
 \models_{\mathit{tot}}\{\ket w\}\ P\ \{\ket v\}
 \quad\Longrightarrow\quad
 \ketbra u=\ketbra v.
\]
Indeed, for the input $\ket w$, the output expectations of $\ket u$ and
$\ket v$ are both at least one, whereas
$\ketbra u+\ketbra v\le(1+|\langle u\mid v\rangle|)I$ and the total
probability of termination is at most one. Hence
$2\le1+|\langle u\mid v\rangle|\le2$.

Applying this fact to~\eqref{eq:qft-circuit-output} and~\eqref{eq:qfttriple}
gives $V\ket{x}=c_x\ket{\chi(n,x)}$ with $|c_x|=1$.
Together with~\eqref{eq:qft-uniform}, it gives
$V\ket{+}^{\otimes n}=c\ket{0}^{\otimes n}$ with $|c|=1$.
The Fourier sum in Example~\ref{ssec:qft} yields
$\ket{0}^{\otimes n}=2^{-n/2}\sum_x\ket{\chi(n,x)}$.
Since $\ket{+}^{\otimes n}=2^{-n/2}\sum_x\ket{x}$, linearity gives
\[
  2^{-n/2}\sum_x c_x\ket{\chi(n,x)}
  =2^{-n/2}\sum_x c\ket{\chi(n,x)}.
\]
The Fourier basis states are orthonormal, so $c_x=c$ for every $x$.
Thus, for every input state $\ket\psi=\sum_x\psi_x\ket x$,
$V\ket\psi=c\sum_x\psi_x\ket{\chi(n,x)}$.
The common global phase does not change its projector, and (Rule-Conseq)
applied to~\eqref{eq:qft-circuit-output} establishes
\begin{equation}\label{eq:qft-all-input}
  \models_{\mathit{tot}}\quad
  \bigl\{\ket\psi\bigr\}\ \QFT_n\
  \Bigl\{\sum_x\psi_x\ket{\chi(n,x)}\Bigr\}.
\end{equation}

\subsection{Bell-State Preparation}\label{ssec:ex-bell}

The Bell state $\ket{\Phi^+}=(\ket{00}+\ket{11})/\sqrt2$ is the shared
resource of quantum teleportation~\cite{Bennett1993} and superdense
coding~\cite{BennettWiesner1992}, and its preparation is the smallest
entangling circuit. The program applies a Hadamard to $q_0$ and then a
$\CNOT$ with control $q_0$ and target $q_1$. Its Hoare triple is
\[
  \models_{\mathit{tot}}
  \bigl\{\, \ket{00} \,\bigr\}\ \
  H[q_0];\ \CNOT[q_0,q_1]\ \
  \bigl\{\, \ket{\Phi^+} \,\bigr\}.
\]

\begin{proofoutline}{fig:belloutline}{Bell-state preparation}
\begin{outlinerows}

\qassert{0}
  {\mathit{true},\ \ket{00}}
  {}

\qline{1}{0}
  {H[q_0]}
  {(Axiom-Uni)}

\qassert{0}
  {\mathit{true},\ \ket{+} \otimes \ket{0}}
  {}

\qline{2}{0}
  {\CNOT[q_0,q_1]}
  {(Axiom-Uni)}

\qassert{0}
  {\mathit{true},\ \ket{\Phi^+}}
  {}

\end{outlinerows}
\end{proofoutline}

The verification conditions at lines 1 and 2 are the gate actions
$H\ket{0} \otimes \ket{0} = \ket{+} \otimes \ket{0}$ and
$\CNOT\,(\ket{+} \otimes \ket{0}) = \ket{\Phi^+}$.

\subsection{Parallel Hadamard}\label{ssec:ex-phwhile}

Parallel Hadamard prepares the uniform superposition $\ket{+}^{\otimes n}$,
the first step of algorithms such as Grover search~\cite{Grover1996} and
phase estimation. The program $\PH_n$ applies a Hadamard to each quantum
variable $q[i]$ of the quantum array $q$, using the loop counter $i$. Its
Hoare triple is, for every $n \geq 1$,
\[
  \models_{\mathit{tot}}
  \bigl\{\, \ket{0}^{\otimes n} \,\bigr\}\ \PH_n\
  \bigl\{\, \ket{+}^{\otimes n} \,\bigr\}.
\]
The key component of the \hyperref[fig:phwhileoutline]{parallel-Hadamard
proof outline} is the loop at line 2, where we use (Rule-Loop-tot). The
loop invariant is $\{\mathit{true},\inv(i)\}$, where
\[
  \inv(i) \;=\; \bigotimes_{j=0}^{n-1}
    \begin{cases} \ket{+} & j < i\\[1pt] \ket{0} & j \geq i \end{cases},
\]
meaning that the first $i$ quantum variables carry $\ket+$.

\begin{proofoutline}{fig:phwhileoutline}{Parallel Hadamard}
\begin{outlinerows}

\qassert{0}
  {\mathit{true},\ \ket{0}^{\otimes n}}
  {}

\qline{1}{0}
  {i := 0}
  {(Axiom-Ass), (Rule-Conseq)}

\qassert{0}
  {i = 0,\ \inv(i)}
  {}

\qline{2}{0}
  {\kw{while}\ i < n\ \kw{do}}
  {(Rule-Loop-tot), $t = \max(n{-}i,0)$}

\qassert{1}
  {i < n,\ \inv(i)}
  {}

\qline{3}{1}
  {H[q[i]]}
  {(Axiom-Uni)}

\qassert{1}
  {i < n,\ \inv(i{+}1)}
  {}

\qline{4}{1}
  {i := i + 1}
  {(Axiom-Ass), (Rule-Conseq)}

\qassert{1}
  {\mathit{true},\ \inv(i)}
  {}

\qline{5}{0}
  {\kw{od}}
  {}

\qassert{0}
  {i \geq n,\ \inv(i)}
  {}

\qassert{0}
  {\mathit{true},\ \ket{+}^{\otimes n}}
  {(Rule-Conseq)}

\end{outlinerows}
\end{proofoutline}

At line 1, $\inv(0)=\ket{0}^{\otimes n}$. At line 3,
$H[q[i]]\,\inv(i)=\inv(i+1)$ for $i<n$. The loop ends at line 5 with
$i\ge n$, which gives $\inv(i)=\ket{+}^{\otimes n}$.

\subsection{Inverse QFT}\label{ssec:ex-iqft}

The inverse QFT is the readout stage of quantum phase estimation: it
carries the Fourier basis state $\ket{\chi(n,x)}$ back to $\ket{x}$. The
program $\IQFT_n$ reverses the QFT circuit of
Section~\ref{ssec:ex-qft}, with the same classical variables: the loop
counters $i,j$ and the classical input $x$. It first runs the swap loop of
the QFT in the same order, since the swaps act on disjoint pairs of
quantum variables. It then runs the iterations of the first $\mathbf{while}$ loop of the QFT in
reverse order with negated phase angles. Iteration $i<n$ of its second
loop reverses iteration $n-1-i$ of the first loop of the QFT,
and its third loop applies
$\cR(-2\pi/2^{i+1-j})[q[n-1-i],q[n-1-j]]$ for $j<i$, which reverses
iteration $i-1-j$ of the second loop of the QFT. Its Hoare triple is
\[
  \models_{\mathit{tot}}
  \bigl\{\, \ket{\chi(n,x)} \,\bigr\}\ \
  \IQFT_n\ \
  \bigl\{\, \ket{x} \,\bigr\}.
\]

The derivation uses the three loops of the
\hyperref[fig:qftoutline]{QFT proof outline} in reverse order. The first
loop, which performs the swaps, uses the loop invariant $\{i\le\lfloor n/2\rfloor,\inv''(i)\}$, where,
in this example,
$\inv''(i)=\bigl(\prod_{m=0}^{i-1}\SWAP[q[m],q[n{-}1{-}m]]\bigr)\ket{\chi(n,x)}$,
meaning that the first $i$ swaps are applied. It carries
$\ket{\chi(n,x)}$ to $\bigotimes_{m=0}^{n-1}\ket{\psi(\theta(m))} = \inv(n)$.
The second loop uses the loop invariant $\{i\le n,\inv(n{-}i)\}$. In its
iteration $i<n$, the third loop uses $\{j\le i,\inv'(n{-}1{-}i,\,i{-}j)\}$:
it starts from $\inv'(n{-}1{-}i,\,i) = \inv(n{-}i)$ and ends at
$\inv'(n{-}1{-}i,\,0)$. The verification conditions are those of lines 3,
6, and 13 of the QFT proof outline, with each gate replaced by its inverse. Each iteration of the third loop
removes one phase term by
\begin{equation}\label{eq:iqft-phase-step}
  \cR(-\alpha)\bigl(\ket{\psi(\beta + b\,\alpha)} \otimes \ket{b}\bigr)
  \;=\; \ket{\psi(\beta)} \otimes \ket{b}.
\end{equation}
The Hadamard then gives $H\ket{\psi(x[k]\,\pi)} = \ket{x[k]}$ at
$k = n{-}1{-}i$, which yields $\inv(n{-}(i{+}1))$. The second loop ends with
$\inv(0)=\ket{x}$.

\subsection{Quantum Phase Estimation}\label{ssec:ex-qpe}

Quantum phase estimation~\cite{Kitaev1995,Cleve1998} estimates an
eigenphase of a unitary and is the core subroutine of Shor's
algorithm~\cite{Shor1997}. This program writes an exactly representable
eigenphase into the control qubits and preserves the target eigenstate. Fix $n\ge1$, a single-qubit
unitary $U$, and a state $\ket{u}$. The classical input $x$
specifies the eigenphase through the classical precondition
\[
  \varphi \;\equiv\;
  U\ket{u} = e^{2\pi i\,\bar{x}/2^{n}}\,\ket{u},
\]
where $\bar{x}$ is as in Example~\ref{ssec:qft}. Thus $\varphi$ requires
$\ket{u}$ to be an eigenstate of $U$ with eigenphase $\bar{x}/2^n$. The
program $\QPE_n$ acts on the $n$ control qubits $q[0], \ldots, q[n{-}1]$
and the target qubit $q[n]$. It prepares the control qubits by $\PH_n$
(Appendix~\ref{ssec:ex-phwhile}), then runs a loop whose iteration $i$
applies the controlled power $\cU^{2^{\,n-1-i}}$ to $q[i], q[n]$, and
finally applies $\IQFT_n$ (Appendix~\ref{ssec:ex-iqft}) to the control
qubits. Its Hoare triple is
\[
  \models_{\mathit{tot}}
  \bigl\{\, \varphi,\ \ket{0}^{\otimes n} \otimes \ket{u} \,\bigr\}\
  \QPE_n\
  \bigl\{\, \ket{x} \otimes \ket{u} \,\bigr\}.
\]

The key component of the \hyperref[fig:qpeoutline]{QPE proof outline} is
the loop at line 3, where we use (Rule-Loop-tot). The loop invariant is
$\{\varphi\land i\le n,\invL(i)\}$, where
\[
 \invL(i)=\bigotimes_{m=0}^{i-1}\ket{\psi(\theta(n-1-m))}
  \otimes\ket{+}^{\otimes(n-i)}\otimes\ket u,
\]
meaning that the first $i$ control qubits carry their phases.

\begin{proofoutline}{fig:qpeoutline}{Quantum phase estimation}
\begin{outlinerows}

\qassert{0}
  {\varphi,\ \ket{0}^{\otimes n} \otimes \ket{u}}
  {}

\qline{1}{0}
  {\PH_n[q[0], \ldots, q[n{-}1]]}
  {proved in Appendix~\ref{ssec:ex-phwhile}}

\qassert{0}
  {\varphi,\ \ket{+}^{\otimes n} \otimes \ket{u}}
  {}

\qline{2}{0}
  {i := 0}
  {(Axiom-Ass), (Rule-Conseq)}

\qassert{0}
  {\varphi\land i = 0,\ \invL(i)}
  {$\invL(0)=\ket{+}^{\otimes n}\otimes\ket u$}

\qline{3}{0}
  {\kw{while}\ i < n\ \kw{do}}
  {(Rule-Loop-tot), $t = n{-}i$}

\qassert{1}
  {\varphi\land i < n,\ \invL(i)}
  {}

\qline{4}{1}
  {\cU^{2^{\,n-1-i}}[q[i], q[n]]}
  {(Axiom-Uni)}

\qassert{1}
  {\varphi\land i < n,\ \invL(i{+}1)}
  {}

\qline{5}{1}
  {i := i + 1}
  {(Axiom-Ass), (Rule-Conseq)}

\qassert{1}
  {\varphi\land i\le n,\ \invL(i)}
  {}

\qline{6}{0}
  {\kw{od}}
  {}

\qassert{0}
  {\varphi\land i = n,\ \invL(n)}
  {}

\qassert{0}
  {\mathit{true},\ \ket{\chi(n,x)} \otimes \ket{u}}
  {(Rule-Conseq)}

\qline{7}{0}
  {\IQFT_n[q[0], \ldots, q[n{-}1]]}
  {proved in Appendix~\ref{ssec:ex-iqft}}

\qassert{0}
  {\mathit{true},\ \ket{x} \otimes \ket{u}}
  {}

\end{outlinerows}
\end{proofoutline}

Lines 1 and 7 use the Hoare triples of Appendices~\ref{ssec:ex-phwhile}
and~\ref{ssec:ex-iqft} on the control qubits, which leave $\varphi$ and the
target qubit unchanged. At line 4, the control $q[i]$ is in
$\ket{+}=\ket{\psi(0)}$, and $\varphi$ gives
\begin{equation}\label{eq:qpe-phase-step}
\begin{aligned}
  \cU^{2^{\,n-1-i}}\bigl(\ket{\psi(0)} \otimes \ket{u}\bigr)
  &\;=\; \ket{\psi\bigl(2^{\,n-1-i} \cdot 2\pi\bar{x}/2^{n}\bigr)}
    \otimes \ket{u}\\
  &\;=\; \ket{\psi(\theta(n{-}1{-}i))} \otimes \ket{u}.
\end{aligned}
\end{equation}
The gate leaves the other control qubits unchanged, which establishes
$\invL(i+1)$. The loop ends at line 6 with
$\invL(n)=\ket{\chi(n,x)}\otimes\ket u$ by~\eqref{eq:fourierket}.

\subsection{Cluster-State Preparation}\label{ssec:ex-cluster}

This example prepares the 1D cluster state, the resource state
of measurement-based quantum computation~\cite{RaussendorfBriegel2001}. The
program $\mathit{Cluster}_n$ has one classical variable, the loop counter
$k$, which each of its two loops initializes. The first loop applies $H$ to
$q[k]$, and the second applies $\CZ$ to $q[k], q[k{+}1]$. Its Hoare triple
is, for every $n \geq 1$,
\[
  \models_{\mathit{tot}}
  \bigl\{\, \ket{0}^{\otimes n} \,\bigr\}\ \mathit{Cluster}_n\
  \bigl\{\, \ket{C_n} \,\bigr\},
\]
where
\[
  \ket{C_n} \;=\; \frac{1}{\sqrt{2^{\,n}}} \sum_{y \in \{0,1\}^n}
    (-1)^{\sum_{j=0}^{n-2} y_j\,y_{j+1}}\,\ket{y}.
\]

The key components of the \hyperref[fig:clusteroutline]{cluster-state proof
outline} are the loops at lines 2 and 7, where we use (Rule-Loop-tot). The
first loop uses the loop invariant $\{k\le n,\inv(k)\}$, with the
verification conditions of Appendix~\ref{ssec:ex-phwhile}. The second uses
$\{k\le n-1,\inv'(k)\}$, where
\[
  \inv'(k) = \frac{1}{\sqrt{2^{\,n}}} \sum_{y \in \{0,1\}^n}
    (-1)^{\sum_{j=0}^{k-1} y_j\,y_{j+1}}\,\ket{y},
\]
meaning that the first $k$ $\CZ$ gates are applied.

\begin{proofoutline}{fig:clusteroutline}{Cluster-state preparation}
\begin{outlinerows}

\qassert{0}
  {\mathit{true},\ \ket{0}^{\otimes n}}
  {}

\qline{1}{0}
  {k := 0}
  {(Axiom-Ass), (Rule-Conseq)}

\qassert{0}
  {k = 0,\ \inv(k)}
  {$\inv(k) = \ket{+}^{\otimes k}
    \otimes \ket{0}^{\otimes (n-k)}$;
    $\inv(0) = \ket{0}^{\otimes n}$}

\qline{2}{0}
  {\kw{while}\ k < n\ \kw{do}}
  {(Rule-Loop-tot), $t = n{-}k$}

\qassert{1}
  {k < n,\ \inv(k)}
  {}

\qline{3}{1}
  {H[q[k]]}
  {(Axiom-Uni)}

\qassert{1}
  {k < n,\ \inv(k{+}1)}
  {}

\qline{4}{1}
  {k := k + 1}
  {(Axiom-Ass), (Rule-Conseq)}

\qassert{1}
  {k\le n,\ \inv(k)}
  {}

\qline{5}{0}
  {\kw{od}}
  {}

\qassert{0}
  {k = n,\ \inv(n)}
  {}

\qassert{0}
  {\mathit{true},\ \inv'(0)}
  {(Rule-Conseq)}

\qline{6}{0}
  {k := 0}
  {(Axiom-Ass), (Rule-Conseq)}

\qassert{0}
  {k = 0,\ \inv'(k)}
  {$\inv'(k)$: the first $k$ $\CZ$ gates applied}

\qline{7}{0}
  {\kw{while}\ k < n{-}1\ \kw{do}}
  {(Rule-Loop-tot), $t = (n{-}1){-}k$}

\qassert{1}
  {k < n{-}1,\ \inv'(k)}
  {}

\qline{8}{1}
  {\CZ[q[k], q[k{+}1]]}
  {(Axiom-Uni)}

\qassert{1}
  {k < n{-}1,\ \inv'(k{+}1)}
  {}

\qline{9}{1}
  {k := k + 1}
  {(Axiom-Ass), (Rule-Conseq)}

\qassert{1}
  {k\le n-1,\ \inv'(k)}
  {}

\qline{10}{0}
  {\kw{od}}
  {}

\qassert{0}
  {k = n{-}1,\ \inv'(n{-}1)}
  {}

\qassert{0}
  {\mathit{true},\ \ket{C_n}}
  {(Rule-Conseq)}

\end{outlinerows}
\end{proofoutline}

By definition, $\inv'(0)=\ket{+}^{\otimes n}$ and
$\inv'(n-1)=\ket{C_n}$. At line 8, $\CZ$ multiplies each computational
basis state $\ket y$ by $(-1)^{y_k y_{k+1}}$, so
$\CZ[q[k],q[k+1]]\,\inv'(k)=\inv'(k+1)$ for $k<n-1$ by linearity.

\subsection{Building Blocks for Fermi--Hubbard Simulation}\label{ssec:ex-fermion}
The Fermi--Hubbard model is a standard model of interacting
electrons~\cite{Hubbard1963} and a benchmark for quantum
simulation~\cite{Wecker2015Hubbard}. We consider its 1D spinless version
on an open chain, with nearest-neighbor hopping and density interactions,
and verify three building blocks for simulating it:
(a) encode the fermionic
states and Hamiltonian into qubits by the Bravyi--Kitaev (BK)
transformation; (b) block-encode the
encoded Hamiltonian by an LCU circuit; and (c) apply a QSP polynomial transformation to that
encoding~\cite{LowChuang2019,GilyenSuLowWiebe2019}. The QSP program composes
them: it runs the BK circuit on its input and calls the LCU circuit in each
iteration.

\subsubsection{The Bravyi--Kitaev Transformation}\label{ssec:ex-bk}

The Bravyi--Kitaev transformation maps fermionic operators to qubit
operators acting on $O(\log n)$ qubits, compared with $O(n)$ for the
Jordan--Wigner transformation~\cite{SeeleyRichardLove2012}. This example
verifies the circuit of the Bravyi--Kitaev transformation, the invertible matrix $\beta_{2^x}$ over $\mathbb{Z}_2$
given by its block recursion:
\[
  \beta_{1} = \begin{pmatrix} 1 \end{pmatrix},
  \qquad
  \beta_{2^{x+1}}
    = \begin{pmatrix} \beta_{2^{x}} & 0 \\[2pt] C_{2^{x}} & \beta_{2^{x}}
      \end{pmatrix},
\]
where $C_{2^x}$ is zero except for its last row, which is all ones.
The circuit computes $\beta_{2^x}f$ in
place~\cite{SeeleyRichardLove2012}, level by level. The program has two
classical variables, the level counter $\ell$ and the block counter $k$.
The first $\mathbf{while}$ loop runs over the $x$ levels, and the second runs over the
blocks of one level. Each iteration of the second loop is a conditional. If block $k$
exists at level $\ell$, that is, $(k{+}1)2^{\ell+1}\le 2^{x}$, it applies a $\CNOT$ with control
$q[\mathit{ctl}(\ell,k)]$ and target $q[\mathit{tgt}(\ell,k)]$, where
$\mathit{ctl}(\ell,k) = k \cdot 2^{\ell+1} + 2^{\ell} - 1$ and
$\mathit{tgt}(\ell,k) = k \cdot 2^{\ell+1} + 2^{\ell+1} - 1$. Otherwise,
it skips. Its Hoare triple is, for every $x$ and every input state
$\ket{\psi}$ of the $2^{x}$ qubits,
\[
  \models_{\mathit{tot}}
  \{\ket{\psi}\}\ \mathit{BK}_x\ \{\beta_{2^x}\ket{\psi}\}.
\]
Here, $\ket{f}$ is the computational basis state of the $2^{x}$ qubits
for $f \in \{0,1\}^{2^{x}}$, and we also write $\beta_{2^x}$ for the unitary
that maps each $\ket{f}$ to $\ket{\beta_{2^x}f}$.

The key components of the \hyperref[fig:bkoutline]{Bravyi--Kitaev proof
outline} are the loops at lines 2 and 4, where we use (Rule-Loop-tot). The
first loop uses the loop invariant $\{\mathit{true},\inv(\ell)\}$, where
$\inv(\ell)$ is the state obtained by applying the first $\ell$ levels of
the circuit to $\ket{\psi}$. The second loop uses $\{\ell < x\wedge k\le 2^x,\inv'(\ell,k)\}$,
where $\inv'(\ell,k)$ is obtained by applying the first $k$ iterations of
level $\ell$ to $\inv(\ell)$.

\begin{proofoutline}{fig:bkoutline}{Bravyi--Kitaev transformation}
\begin{outlinerows}

\qassert{0}
  {\mathit{true},\ \ket{\psi}}
  {}

\qline{1}{0}
  {\ell := 0}
  {(Axiom-Ass), (Rule-Conseq)}

\qassert{0}
  {\ell = 0,\ \inv(\ell)}
  {$\inv(0) = \ket{\psi}$}

\qline{2}{0}
  {\kw{while}\ \ell < x\ \kw{do}}
  {(Rule-Loop-tot), $t = \max(x{-}\ell,0)$}

\qassert{1}
  {\ell < x,\ \inv(\ell)}
  {}

\qline{3}{1}
  {k := 0}
  {(Axiom-Ass), (Rule-Conseq)}

\qassert{1}
  {\ell < x \wedge k = 0,\ \inv'(\ell,k)}
  {$\inv'(\ell,0)=\inv(\ell)$}

\qline{4}{1}
  {\kw{while}\ k < 2^{x}\ \kw{do}}
  {(Rule-Loop-tot), $t = 2^{x}{-}k$}

\qassert{2}
  {\ell < x \wedge k < 2^{x},\ \inv'(\ell,k)}
  {}

\qline{5}{2}
  {\kw{if}\ (k{+}1) 2^{\ell+1} \le 2^{x}}
  {(Rule-Cond)}

\qassert{2}
  {\ell < x \wedge k < 2^{x} \wedge
    (k{+}1) 2^{\ell+1} \le 2^{x},\ \inv'(\ell,k)}
  {first premise}

\qline{6}{2}
  {\kw{then}\ \CNOT[q[\mathit{ctl}], q[\mathit{tgt}]]}
  {(Axiom-Uni)}

\qassert{2}
  {\ell < x \wedge k < 2^{x} \wedge
    (k{+}1) 2^{\ell+1} \le 2^{x},\ \inv'(\ell,k+1)}
  {}

\qassert{2}
  {\ell < x \wedge k < 2^{x} \wedge
    (k{+}1) 2^{\ell+1} > 2^{x},\ \inv'(\ell,k)}
  {second premise}

\qline{7}{2}
  {\kw{else}\ \kw{skip}}
  {(Axiom-Ski), (Rule-Conseq)}

\qassert{2}
  {\ell < x \wedge k < 2^{x} \wedge
    (k{+}1) 2^{\ell+1} > 2^{x},\ \inv'(\ell,k+1)}
  {}

\qassert{2}
  {\ell < x \wedge k < 2^{x},\ \inv'(\ell,k+1)}
  {}

\qline{8}{2}
  {k := k + 1}
  {(Axiom-Ass), (Rule-Conseq)}

\qassert{2}
  {\ell < x\wedge k\le 2^x,\ \inv'(\ell,k)}
  {}

\qline{9}{1}
  {\kw{od}}
  {}

\qassert{1}
  {\ell < x \wedge k = 2^{x},\ \inv'(\ell,2^{x})}
  {}

\qassert{1}
  {\ell < x,\ \inv(\ell{+}1)}
  {(Rule-Conseq)}

\qline{10}{1}
  {\ell := \ell + 1}
  {(Axiom-Ass), (Rule-Conseq)}

\qassert{1}
  {\mathit{true},\ \inv(\ell)}
  {}

\qline{11}{0}
  {\kw{od}}
  {}

\qassert{0}
  {\ell \geq x,\ \inv(\ell)}
  {}

\qassert{0}
  {\mathit{true},\ \beta_{2^x}\ket{\psi}}
  {(Rule-Conseq)}

\end{outlinerows}
\end{proofoutline}

At line 6, the $\CNOT$ carries $\inv'(\ell,k)$ to $\inv'(\ell,k+1)$. At
line 7, block $k$ does not exist, so $\inv'(\ell,k+1)=\inv'(\ell,k)$. The
second loop ends with $\inv'(\ell,2^{x})=\inv(\ell+1)$. After line 11,
$\ell\ge x$, and $\inv(\ell)=\inv(x)$ because no block exists at a level
$\ell\ge x$. Each $\CNOT$ updates one bit of a computational basis state,
so the first $\ell$ levels map $\ket{f}$ to $\ket{\mathrm{net}_{\ell}\, f}$,
where $\mathrm{net}_{\ell}\, f$ is obtained from $f$ by the bit updates of
the first $\ell$ levels. Finally, $\mathrm{net}_{x}\, f=\beta_{2^x} f$
states that the $x$ levels compute $\beta_{2^x}$, so
$\inv(x)=\beta_{2^x}\ket{\psi}$ by linearity.

\subsubsection{The Fermi--Hubbard LCU Block Encoding}\label{ssec:ex-lcu}

Block encoding by a linear combination of unitaries (LCU) embeds a
Hamiltonian, scaled by $1/\lambda$, in a block of a unitary, and is the
input model of quantum signal
processing. This example
block-encodes the Fermi--Hubbard Hamiltonian by a program with nested
loops and classically computed angles and signs. The fermionic operators
$a_j,a_j^\dagger$ are linear combinations of Pauli strings under the BK
encoding, and we prove that they satisfy the fermionic anticommutation
relations. For $n=2^x$, the Hamiltonian is
\[
  H \;=\; -t_{\mathrm h} \sum_{j=0}^{n-2}\bigl(a^{\dagger}_j a_{j+1}
      + a^{\dagger}_{j+1} a_j\bigr)
    \;+\; U \sum_{j=0}^{n-2} n_j\, n_{j+1}
  \;=\; \sum_{l=0}^{N-1} c_l\, P_l,
  \qquad n_j = a^{\dagger}_j a_j,
\]
with $N = 6(n-1)$ Pauli strings $P_l$, six per adjacent pair, and real
coefficients $c_l$ computed from $(t_{\mathrm h}, U)$. Let
\begin{align*}
 \lambda &= \sum_{l=0}^{N-1}|c_l|, &
 \omega_l &= \begin{cases}
   \sqrt{|c_l|/\lambda},&\lambda>0,\\
   \delta_{l,l_0},&\lambda=0,
 \end{cases}\\
 \tau_j &= \Bigl(\sum_{l=j}^{N-1}\omega_l^2\Bigr)^{1/2}, &
 \theta_j &= \begin{cases}
   \arccos(\omega_j/\tau_j),&\tau_j>0,\\
   \pi/2,&\tau_j=0,
 \end{cases}
\end{align*}
where $l_0$ is a fixed index, and let $\mathit{sg}_l=1$ for $c_l\ge0$ and
$\mathit{sg}_l=-1$ for $c_l<0$. Let $G(\theta)$ be the two-qubit rotation
by $\theta$ in the subspace spanned by $\ket{01}$ and $\ket{10}$,
\[
  G(\theta)\ket{10} = \cos\theta\,\ket{10} + \sin\theta\,\ket{01},
  \qquad
  G(\theta)\ket{01} = -\sin\theta\,\ket{10} + \cos\theta\,\ket{01},
\]
which acts as the identity on $\ket{00}$ and $\ket{11}$.
The sign gate on $a[l]$ has matrix $\operatorname{diag}(1,\mathit{sg}_l)$.

The program acts on $n$ system qubits $q[s]$ and $N$ ancilla qubits
$a[l]$, one per Hamiltonian term. We write $\ket{1_l}_a$ for the basis
state with ancilla qubit $l$ in $\ket{1}$ and every other ancilla qubit in
$\ket{0}$. The classical variables are the loop counters and the
Hamiltonian coefficients $(t_{\mathrm h}, U)$.

The program begins with the ancilla initialization
$\bar{a} := \ket{0}$, where $\bar{a}$ is the sequence of all $N$
ancilla qubits. It continues with the circuit $\LCU$, which has the
standard prepare--select--unprepare structure. Line 1 applies $X$ to
$a[0]$, and the first $\mathbf{while}$ loop (lines 2--6) applies the
rotations $G(\theta_j(\sigma))$ to $a[j], a[j{+}1]$, which carries
$\ket{1_0}_a$ to the state with amplitude $\omega_l(\sigma)$ on each
$\ket{1_l}_a$. The second loop (lines 7--12) runs over the $N$ ancilla
qubits. Its iteration $l$ applies the sign gate to $a[l]$ and then,
controlled on $a[l]$, the single-qubit factors of $P_l$ to the system
qubits, one factor per iteration of the third loop. The fourth loop
(lines 13--17) applies the inverse rotations
$G(-\theta_{N-2-j}(\sigma))$ at the reversed positions, and line 18
applies $X$ to $a[0]$ again.

The four loops use the loop invariants $\{j\le N-1,\inv_{(0)}(j)\}$,
$\{l\le N,\inv_{(1)}(l)\}$, $\{l<N\wedge s\le n,\inv_{(2)}(l,s)\}$, and
$\{j\le N-1,\inv_{(3)}(j)\}$, where
\begin{align*}
 \inv_{(0)}(j)&=\textstyle\bigl(\sum_{i<j}\omega_i\ket{1_i}_a
   +\tau_j\ket{1_j}_a\bigr)\otimes\ket\psi,
 \qquad
 \inv_{(3)}(j)=G^{-1}_{<j}\textstyle\sum_{l}
 \omega_l\mathit{sg}_l\ket{1_l}_a\otimes P_l\ket\psi,\\
 \inv_{(1)}(l)&=\textstyle\sum_{i<l}\omega_i\mathit{sg}_i\ket{1_i}_a\otimes P_i\ket\psi
   +\sum_{i\ge l}\omega_i\ket{1_i}_a\otimes\ket\psi,\\
 \inv_{(2)}(l,s)&=\textstyle\sum_{i<l}\omega_i\mathit{sg}_i\ket{1_i}_a\otimes P_i\ket\psi
 +\omega_l\mathit{sg}_l\ket{1_l}_a\otimes P_{l,<s}\ket\psi
 +\sum_{i>l}\omega_i\ket{1_i}_a\otimes\ket\psi.
\end{align*}
Here, $P_{l,<s}$ is the product of the first $s$ single-qubit factors of
$P_l$, $(P_l)_s$ is the factor acting on $q[s]$, and $G^{-1}_{<j}$ is the
product of the first $j$ inverse rotations of the fourth loop. The
invariant $\inv_{(0)}(j)$ means that the ancilla amplitudes
$\omega_0,\ldots,\omega_{j-1}$ are prepared, and $\inv_{(1)}(l)$ means
that the signed Pauli strings of the first $l$ terms are applied. The
\hyperref[fig:lcuoutline]{LCU proof outline} derives
\[
  \models_{\mathit{tot}}
  \bigl\{\, \ket{0\cdots0}_a \otimes \ket{\psi} \,\bigr\}\ \
  \LCU\ \
  \bigl\{\, \ket{\Phi_{\mathrm{out}}(\sigma)} \,\bigr\},
  \qquad
  \Phi_{\mathrm{out}}(\sigma)=X_{a[0]}\,\inv_{(3)}(N{-}1),
\]
for every $x > 0$ and every input state $\ket\psi$ of the system qubits.

\begin{proofoutline}{fig:lcuoutline}{LCU block encoding}
\begin{outlinerows}[\scriptsize\setlength{\tabcolsep}{2pt}]

\qassert{0}
  {\mathit{true},\ \ket{0\cdots0}_a \otimes \ket{\psi}}
  {}

\qline{1}{0}
  {X[a[0]]}
  {(Axiom-Uni)}

\qassert{0}
  {\mathit{true},\ \ket{1_0}_a \otimes \ket{\psi}}
  {}

\qline{2}{0}
  {j := 0}
  {(Axiom-Ass), (Rule-Conseq)}

\qassert{0}
  {j = 0,\ \inv_{(0)}(j)}
  {$\tau_0=1$}

\qline{3}{0}
  {\kw{while}\ j < N{-}1\ \kw{do}}
  {(Rule-Loop-tot), $t = (N{-}1){-}j$}

\qassert{1}
  {j < N{-}1,\ \inv_{(0)}(j)}
  {}

\qline{4}{1}
  {G(\theta_j(\sigma))[a[j], a[j{+}1]]}
  {(Axiom-Uni)}

\qassert{1}
  {j < N{-}1,\ \inv_{(0)}(j{+}1)}
  {}

\qline{5}{1}
  {j := j + 1}
  {(Axiom-Ass), (Rule-Conseq)}

\qassert{1}
  {j\le N-1,\ \inv_{(0)}(j)}
  {}

\qline{6}{0}
  {\kw{od}}
  {}

\qassert{0}
  {j = N{-}1,\ \inv_{(0)}(N{-}1)}
  {}

\qassert{0}
  {\mathit{true},\ \textstyle\sum_{l=0}^{N-1} \omega_l(\sigma)\,\ket{1_l}_a \otimes \ket{\psi}}
  {(Rule-Conseq)}

\qline{7}{0}
  {l := 0}
  {(Axiom-Ass), (Rule-Conseq)}

\qassert{0}
  {l = 0,\ \inv_{(1)}(l)}
  {}

\qline{8}{0}
  {\kw{while}\ l < N\ \kw{do}}
  {(Rule-Loop-tot), $t = N{-}l$}

\qassert{1}
  {l < N,\ \inv_{(1)}(l)}
  {}

\qline{9}{1}
  {\mathit{sg}_l(\sigma)[a[l]]}
  {(Axiom-Uni)}

\qassert{1}
  {l < N,\ \inv_{(2)}(l,0)}
  {}

\qline{10}{1}
  {s := 0;\ \kw{while}\ s < n\ \kw{do}}
  {(Rule-Loop-tot), $t = n{-}s$, $\inv_{(2)}(l,s)$}

\qline{}{2}
  {\ctrl{(P_l)_s}[a[l], q[s]];\ s := s{+}1\ \kw{od}}
  {}

\qassert{1}
  {l < N,\ \inv_{(1)}(l{+}1)}
  {}

\qline{11}{1}
  {l := l + 1}
  {(Axiom-Ass), (Rule-Conseq)}

\qassert{1}
  {l\le N,\ \inv_{(1)}(l)}
  {}

\qline{12}{0}
  {\kw{od}}
  {}

\qassert{0}
  {l = N,\ \inv_{(1)}(N)}
  {}

\qassert{0}
  {\mathit{true},\ \textstyle\sum_{l=0}^{N-1} \omega_l(\sigma)\,\mathit{sg}_l(\sigma)\,\ket{1_l}_a \otimes P_l\ket{\psi}}
  {(Rule-Conseq)}

\qline{13}{0}
  {j := 0}
  {(Axiom-Ass), (Rule-Conseq)}

\qassert{0}
  {j = 0,\ \inv_{(3)}(j)}
  {}

\qline{14}{0}
  {\kw{while}\ j < N{-}1\ \kw{do}}
  {(Rule-Loop-tot), $t = (N{-}1){-}j$}

\qassert{1}
  {j < N{-}1,\ \inv_{(3)}(j)}
  {}

\qline{15}{1}
  {G(-\theta_{N-2-j}(\sigma))[a[N{-}2{-}j], a[N{-}1{-}j]]}
  {(Axiom-Uni)}

\qassert{1}
  {j < N{-}1,\ \inv_{(3)}(j{+}1)}
  {}

\qline{16}{1}
  {j := j + 1}
  {(Axiom-Ass), (Rule-Conseq)}

\qassert{1}
  {j\le N-1,\ \inv_{(3)}(j)}
  {}

\qline{17}{0}
  {\kw{od}}
  {}

\qassert{0}
  {j = N{-}1,\ \inv_{(3)}(N{-}1)}
  {}

\qassert{0}
  {\mathit{true},\ \inv_{(3)}(N{-}1)}
  {(Rule-Conseq)}

\qline{18}{0}
  {X[a[0]]}
  {(Axiom-Uni)}

\qassert{0}
  {\mathit{true},\ \ket{\Phi_{\mathrm{out}}(\sigma)}}
  {}

\end{outlinerows}
\end{proofoutline}

The verification condition for the rotation at line 4 is
\begin{equation}\label{eq:lcu-rotation-step}
  G(\theta_j(\sigma))[a[j],a[j{+}1]]\,\bigl(\tau_j\,\ket{1_j}_a\bigr)
  \;=\; \omega_j\,\ket{1_j}_a + \tau_{j+1}\,\ket{1_{j+1}}_a,
\end{equation}
which carries $\inv_{(0)}(j)$ to $\inv_{(0)}(j{+}1)$. Since the fourth loop
applies the inverses of the rotations of the first loop in reverse order,
$(\bra{1_0}_a\otimes I)\,G^{-1}_{<N-1}=\sum_{l}\omega_l(\sigma)
(\bra{1_l}_a\otimes I)$. Together with $\inv_{(3)}(0)=\inv_{(1)}(N)$ and
$\omega_l^2\,\mathit{sg}_l=c_l/\lambda$, this gives, for $\lambda(\sigma)>0$,
the block-encoding identity
\begin{equation}\label{eq:lcu-block-target}
  \bigl(\bra{0\cdots0}_a \otimes I\bigr)\,
  \ket{\Phi_{\mathrm{out}}(\sigma)}
  \;=\; \lambda(\sigma)^{-1}\, H\ket{\psi}.
\end{equation}

For every quantum predicate $K$ on the system qubits, $0\le K\le I$, and
every $x>0$, we derive
\begin{equation}\label{eq:lcu-predicate-triple}
\begin{aligned}
  \models_{\mathit{tot}}\
  &\bigl\{\, \lambda(\sigma)>0,\
    I_a\otimes\lambda(\sigma)^{-2}H^\dagger K H\,\bigr\}\
  \bar a:=\ket0;\ \LCU\\
  &\bigl\{\,\ketbra{0\cdots0}_a\otimes K\,\bigr\}.
\end{aligned}
\end{equation}

We derive this Hoare triple backward from the postcondition
$A=\ketbra{0\cdots0}_a\otimes K$. Each gate and the initialization use the
precondition given by their axioms. A loop with bound $M$ uses the loop
invariant $W_j^\dagger C W_j$, where $C$ is the quantum predicate required
after the loop and $W_j$ is the unitary applied by its iterations
$j,\ldots,M-1$, which is $I$ for $j\ge M$. Its verification conditions are
\[
 W_j=I\ \text{ for } j\ge M,\qquad
 U_j^\dagger(W_{j+1}^\dagger C W_{j+1})U_j=W_j^\dagger C W_j\ \text{ for } j<M,
\]
where $U_j$ is the unitary applied by iteration $j$. The third loop is
derived in the same way. For an initialization or a circuit $P$ and a
quantum predicate $A$, we write $\operatorname{wp}(P,A)$ for the
precondition computed from $A$ by the axioms (Section~\ref{ssec:lean-rules})
and these loop invariants. The Hoare
triple~\eqref{eq:lcu-predicate-triple} then follows from (Axiom-Init),
the \hyperref[fig:lcuoutline]{LCU proof outline} with these loop
invariants, and (Rule-Conseq), whose premise
\[
 I_a\otimes\lambda(\sigma)^{-2}H^\dagger K H
 \le \operatorname{wp}(\bar a:=\ket0;\LCU,A)(\sigma)
 \qquad(\lambda(\sigma)>0)
\]
holds with equality by the block-encoding
identity~\eqref{eq:lcu-block-target}. Hence the precondition
of~\eqref{eq:lcu-predicate-triple} equals
$\operatorname{wp}(\bar a:=\ket0;\LCU,A)$. The same derivation
applies to $I-A$. Since the total output trace is at most the input trace,
the expectation of $A$ in the output equals that of the precondition in
the input. Taking $K=I$ gives the
success probability $\lambda^{-2}\|H\ket\psi\|^2$.
As $K$ ranges over all quantum predicates, the precondition determines
the successful output $\lambda^{-2}H\ketbra{\psi}H^\dagger$.

\subsubsection{The QSP Polynomial Transformation}\label{ssec:ex-qsp}

Quantum signal processing (QSP) applies polynomial transformations to a
block-encoded operator and underlies efficient Hamiltonian
simulation~\cite{LowChuang2019,GilyenSuLowWiebe2019}. The QSP program
alternates the LCU circuit of
Appendix~\ref{ssec:ex-lcu} with phase gates and reflections on the ancilla
qubit $a[0]$. After the ancilla initialization
$\bar{a} := \ket{0}$, the program $\QSP_d$ applies $\mathit{BK}_x$
(Appendix~\ref{ssec:ex-bk}) to the system qubits and then runs a loop of
depth $d$ on the counter $k$. Iteration $k$ applies the phase gate
$\phasegate(\Phi[k])[a[0]]$, with the angle $\Phi[k]$ given by the
classical array $\Phi$, then $\LCU$, then $Z[a[0]]$, and increments $k$.

Let $e_\psi=\ket{0\cdots0}_a\otimes\beta_{2^x}\ket\psi$, and let $U_{\LCU}$, $\Ph(\alpha)$,
and $R_Z$ be the unitaries applied by $\LCU$, $\phasegate(\alpha)[a[0]]$,
and $Z[a[0]]$; the single-qubit matrix of $\phasegate(\alpha)$ is
$\operatorname{diag}(1,e^{i\alpha})$. Let
$\Pi_m=\prod_{k=0}^{m-1}R_Z\,U_{\LCU}\,\Ph(\Phi[k])$~\cite{GilyenSuLowWiebe2019}, with the $k=0$
factor acting first. The \hyperref[fig:qspoutline]{QSP proof outline}
derives
\[
  \models_{\mathit{tot}}
  \bigl\{\, \ket{0\cdots0}_a \otimes \ket{\psi} \,\bigr\}\ \
  \QSP_d\ \
  \bigl\{\, \ket{\Xi_{\mathrm{out}}(\sigma)} \,\bigr\},
  \qquad
  \Xi_{\mathrm{out}}(\sigma) \;=\; \Pi_d\, e_{\psi},
\]
for every $x>0$, depth $d$, and input state $\ket\psi$. The key component
of the outline is the loop at line 3, where we use (Rule-Loop-tot). The
loop invariant is $\{k\le d,\invL(k)\}$, where $\invL(k)=\Pi_k\,e_\psi$.

\begin{proofoutline}{fig:qspoutline}{QSP polynomial transformation}
\begin{outlinerows}

\qassert{0}
  {\mathit{true},\ \ket{0\cdots0}_a \otimes \ket{\psi}}
  {}

\qline{1}{0}
  {\mathit{BK}_x[q[0], \ldots, q[n{-}1]]}
  {proved in Appendix~\ref{ssec:ex-bk}}

\qassert{0}
  {\mathit{true},\ e_{\psi}}
  {}

\qline{2}{0}
  {k := 0}
  {(Axiom-Ass), (Rule-Conseq)}

\qassert{0}
  {k = 0,\ \invL(k)}
  {$\Pi_0=I$}

\qline{3}{0}
  {\kw{while}\ k < d\ \kw{do}}
  {(Rule-Loop-tot), $t = d{-}k$}

\qassert{1}
  {k < d,\ \invL(k)}
  {}

\qline{4}{1}
  {\phasegate(\Phi[k])[a[0]]}
  {(Axiom-Uni)}

\qassert{1}
  {k < d,\ \Ph(\Phi[k])\,\Pi_k\, e_{\psi}}
  {}

\qline{5}{1}
  {\LCU}
  {as in the \hyperref[fig:lcuoutline]{LCU proof outline}}

\qassert{1}
  {k < d,\ U_{\LCU}\,\Ph(\Phi[k])\,\Pi_k\, e_{\psi}}
  {}

\qline{6}{1}
  {Z[a[0]]}
  {(Axiom-Uni)}

\qassert{1}
  {k < d,\ \invL(k{+}1)}
  {}

\qline{7}{1}
  {k := k + 1}
  {(Axiom-Ass), (Rule-Conseq)}

\qassert{1}
  {k\le d,\ \invL(k)}
  {}

\qline{8}{0}
  {\kw{od}}
  {}

\qassert{0}
  {k = d,\ \invL(k)}
  {}

\qassert{0}
  {\mathit{true},\ \Pi_d\, e_{\psi}}
  {(Rule-Conseq)}

\end{outlinerows}
\end{proofoutline}

Line 1 uses the Hoare triple of Appendix~\ref{ssec:ex-bk} on the system
qubits, which leaves the ancilla qubits unchanged. At line 5, the loops of
$\LCU$ use as loop invariants the states obtained
by applying the gates of each loop up to the current iteration to
$\Ph(\Phi[k])\,\Pi_k\,e_\psi$. The verification conditions hold by the
definition of $\Pi_k$: $\Pi_0\,e_\psi=e_\psi$, each iteration maps
$\Pi_k\,e_\psi$ to $R_Z\,U_{\LCU}\,\Ph(\Phi[k])\,\Pi_k\,e_\psi=\Pi_{k+1}\,e_\psi$,
and the loop ends with $\Pi_d\,e_\psi=\Xi_{\mathrm{out}}(\sigma)$.

Let $\ket v$ be an eigenstate with $H\ket v=E\ket v$, and assume $\lambda>0$
and $|E|<\lambda$. We prove that, on the states the loop reaches from
$\ket{0\cdots0}_a\otimes\ket v$, $R_Z$ is the reflection about the subspace
with all ancilla qubits in $\ket0$, and that the loop preserves a
two-dimensional subspace containing
$\ket{0\cdots0}_a\otimes\ket v$~\cite{GilyenSuLowWiebe2019}. In a basis of
this subspace that starts with $\ket{0\cdots0}_a\otimes\ket v$, the
upper-left entry of $\Pi_d$ is, with $c=E/\lambda$,
\begin{equation}\label{eq:qsp-polynomial-entry}
 P_{\Phi}(c)=\left[
  \prod_{k=0}^{d-1}
  \begin{pmatrix}
   c & e^{i\Phi[k]}\sqrt{1-c^2}\\
   -\sqrt{1-c^2} & e^{i\Phi[k]}c
  \end{pmatrix}
 \right]_{00},
\end{equation}
with the $k=0$ factor acting first, as in $\Pi_d$. Hence
\[
 (\bra{0\cdots0}_a\otimes I)\,\Pi_d\,(\ket{0\cdots0}_a\otimes\ket v)
 =P_\Phi(E/\lambda)\ket v.
\]

For arbitrary inputs, assume $\lambda>0$ and $|E|<\lambda$ for every
eigenvalue $E$ of $H$, and define $P_\Phi(H/\lambda)$ by
$\ket{v}\mapsto P_\Phi(E/\lambda)\ket{v}$ for each eigenstate $\ket v$ with
$H\ket v=E\ket v$. By linearity, the eigenstate
identity extends to
\[
 (\bra{0\cdots0}_a\otimes I)\Pi_d
 (\ket{0\cdots0}_a\otimes I)=P_\Phi(H/\lambda).
\]
Including the initial BK transformation, the block is therefore
\[
 (\bra{0\cdots0}_a\otimes I)\Pi_d
 (I_a\otimes\beta_{2^x})(\ket{0\cdots0}_a\otimes I)
 =P_\Phi(H/\lambda)\beta_{2^x}.
\]
Let
$\varphi_{\mathrm{sp}}$ denote these spectral hypotheses. They hold, for
example, when $x=1$ and $t_h=U=1$. For every quantum predicate $K$ on the
system qubits, $0\le K\le I$, we derive
\begin{equation}\label{eq:qsp-predicate-triple}
\begin{aligned}
 \models_{\mathit{tot}}\
 &\{\varphi_{\mathrm{sp}},\
 I_a\otimes \beta_{2^x}^\dagger P_\Phi(H/\lambda)^\dagger K P_\Phi(H/\lambda)\,\beta_{2^x}\}\
 \bar a:=\ket0;\ \QSP_d\\
 &\{\ketbra{0\cdots0}_a\otimes K\}.
\end{aligned}
\end{equation}
We derive it in the same way as~\eqref{eq:lcu-predicate-triple}, with
$\QSP_d$ in place of $\LCU$ and $A=\ketbra{0\cdots0}_a\otimes K$.
Its loop invariant is $W_k^\dagger A W_k$, where
$W_k=\prod_{r=k}^{d-1}R_ZU_{\LCU}\Ph(\Phi[r])$, with the $r=k$ factor
acting first, and the bound function is $\max(d-k,0)$. The initial (Rule-Conseq)
premise holds with equality:
\[
 \operatorname{wp}(\bar a:=\ket0;\QSP_d,A)
 =I_a\otimes \beta_{2^x}^\dagger P_\Phi(H/\lambda)^\dagger K P_\Phi(H/\lambda)\,\beta_{2^x}
 \qquad(\varphi_{\mathrm{sp}}).
\]
Methods for choosing $\Phi$ so that $P_\Phi$ approximates a target
polynomial are given in~\cite{LowChuang2019,GilyenSuLowWiebe2019}.

\end{document}